\documentclass[preprint,12pt]{elsarticle}

\usepackage{graphicx}
\pdfoutput=1
\usepackage{amsmath}
\usepackage{amssymb}
\usepackage{multirow}
\usepackage{pbox}
\usepackage{subfigure}
\usepackage{color}

\usepackage{numcompress}
\biboptions{sort&compress}

\newtheorem{assumption}{Assumption}
\newtheorem{lemma}{Lemma}
\newtheorem{remark}{Remark}

\newtheorem{theorem}{Theorem}
\newenvironment{proof}{{\indent \indent \it Proof:}}

\begin{document}

\begin{frontmatter}
\title{	Adaptive Smooth Disturbance Observer-Based Fast Finite-Time Attitude Tracking Control of a Small Unmanned Helicopter \tnoteref{footnoteinfo}}

\tnotetext[footnoteinfo]{\textcolor{red}{This is the preprint version of the accepted Manuscript: Xidong Wang, Zhan Li, Xinghu Yu, Zhen He, ``Adaptive Smooth Disturbance Observer-Based Fast Finite-Time Attitude Tracking Control of a Small Unmanned Helicopter", Journal of the Franklin Institute, 2022, ISSN: 0016-0032, DOI: 10.1016/j.jfranklin.2022.05.035. Please cite the publisher's version}}
\author[mymainaddress]{Xidong Wang}\ead{17b904039@stu.hit.edu.cn}

\author[mymainaddress,mysecondaryaddress]{Zhan Li\corref{mycorrespondingauthor}}
\cortext[mycorrespondingauthor]{Corresponding author}
\ead{zhanli@hit.edu.cn}

\author[myfirstaddress]{Xinghu Yu}\ead{17b304003@stu.hit.edu.cn}

\author[mythirdaddress]{Zhen He}\ead{hezhen@hit.edu.cn}

\address[mymainaddress]{the Research Institute of Intelligent Control and Systems, Harbin Institute of Technology, Harbin 150001, China}
\address[mysecondaryaddress]{Peng Cheng Laboratory, Shenzhen 518055, China}
\address[myfirstaddress]{Ningbo Institute of Intelligent Equipment Technology Co., Ltd., Ningbo 315201, China}
\address[mythirdaddress]{the Department of Control Science and Engineering, Harbin Institute of Technology, Harbin 150001, China}

\begin{abstract}
In this paper, a novel control scheme, which includes an improved fast finite-time adaptive backstepping controller based on a newly designed adaptive smooth disturbance observer, is presented for the attitude tracking of a small unmanned helicopter system subject to lumped disturbances. First, an adaptive smooth disturbance observer (ASDO) is proposed to approximate the lumped disturbances, which owns the characteristics of smooth output, fast finite-time convergence, and adaptability to disturbances. Then, a finite-time backstepping control protocol is constructed to achieve the attitude tracking objective. By virtue of our newly proposed inequality, a singularity-free virtual control law is designed. To tackle the ``explosion of complexity" problem, a fast finite-time command filter (FFTCF) is utilized to estimate the virtual control signals and their derivatives. In addition, an auxiliary dynamic system is introduced to attenuate the impact of the errors caused by FFTCF estimation. Moreover, an adaptive law with $\sigma $-modification term is developed to compensate the ASDO approximation error.  Theoretical analysis proves that all signals of the closed-loop system are fast finite-time bounded, while the attitude tracking errors fast converge to a small region of the origin in finite time. Finally, two contrastive numerical simulations are carried out to validate the effectiveness and superiority of the designed control scheme.
\end{abstract}


\end{frontmatter}

\section{Introduction}

Due to the merits of vertical take-off and landing, hovering, as well as aggressive maneuverability, the small unmanned helicopter has broad applications both in military and civil fields \cite{2018helicopter,2019helicopter,2021helicopter1,2021helicopter2}. However, small unmanned helicopters have the characteristics of high nonlinearity, strong coupling, under-actuated and are vulnerable to lumped disturbances during flight, which make it a challenge to design the high-performance attitude tracking controller \cite{2.Li2015F}. 

In recent years, researchers have presented numerous approaches to enhance the attitude tracking performance of helicopters. Some linear control methods were adopted to achieve the stability of the helicopter system, including LQR control \cite{4.liu2013L} and H$\infty$ control \cite{5.2013Robust}. The design of these linear controllers relies on the linearization of the system model, which performs well in a small neighborhood of the equilibrium point. Due to the high nonlinearity of the helicopter system, the performance of these linear controllers degrades when far away from the equilibrium point.

To cope with the nonlinearities existing in the system, a variety of nonlinear control and intelligent control methods are employed in the attitude tracking control of small unmanned helicopters. By utilizing a robust compensator to identify uncertainties, a robust hierarchical controller was designed for the desired tracking of a helicopter platform \cite{6.liu2014}. In \cite{2.Li2015F}, a nonlinear robust controller was proposed to implement the semi-global asymptotic attitude tracking of this platform, which embraces an auxiliary system to generate filtered error signals and an uncertainty and disturbance estimator to compensate the unknown lumped perturbations. In \cite{7.zeghlache2017}, a novel sliding mode control (SMC) scheme, combined with an interval type-2 fuzzy logic control approach, was presented to guarantee the exponential tracking error convergence of the helicopter system, which reduced the chattering effect of SMC as well as the rules number of fuzzy controller simultaneously. In \cite{8.Chen2018N}, an adaptive neural network control law was constructed for the helicopter platform with the aid of backstepping technique, while the neural network was designed to identify uncertainties. The authors of \cite{9.liu2019attitude} designed a new nonlinear attitude tracking controller with disturbance compensation to achieve asymptotical stability of the helicopter system. In \cite{10.Yang2020F}, a RBFNN-based backstepping control strategy was proposed for the experimental helicopter, where the RBFNN was utilized to estimate lumped disturbances. The paper \cite{11.li2021} presented a NN backstepping control approach integrated with command filtering to investigate the tracking control problem of the helicopter platform. Furthermore, the fault-tolerant control of the helicopter system was studied in \cite{12.chen2016N,13.Zeghlache2017N,14.perez2020fault}. Most of the control approaches mentioned above are asymptotically stable or ultimately uniformly bounded, which means that the signals of closed-loop system converge to the equilibrium point or its neighborhood in infinite time. However, the finite-time stability of the helicopter platform is seldom considered in the literature, which is more desirable in practical applications.

As one of the effective approaches to implement finite-time control, the finite-time backstepping technique has sparked much attention because of its powerful capability to achieve finite-time convergence while maintaining the performance of traditional backstepping control \cite{lemma1}. Thereafter, a lot of work has been done to enhance the performance of the finite-time backstepping control \cite{2019fnt1,2020fnt1,2020fnt2,2020fnt3,2021fnt1,2021fnt2,2021fnt3,2021fnt4}. The authors of \cite{2021fnt1} designed a novel fractional power-based auxiliary dynamic system to replace the original sign function-based one, which effectively improves the  error compensation performance.

Motivated by the above analysis, in this paper, we propose a novel adaptive smooth disturbance observer-based fast finite-time adaptive backstepping control strategy for the attitude tracking of the small unmanned helicopter system subject to lumped disturbances. First, a novel ASDO is developed to estimate the lumped disturbance of helicopter system control channel. Then, an improved finite-time backstepping control scheme is constructed to achieve the attitude tracking objective. In addition, a FFTCF is utilized to approximate the virtual control signals and their derivatives, while an auxiliary dynamic system \cite{2021fnt1} is introduced to compensate the errors caused by FFTCF approximation. Moreover, an adaptive law with $\sigma $-modification term is designed to weaken the ASDO estimation error effect. By utilizing the fast finite-time stability theory, it is proved that all the closed-loop system signals are fast finite-time bounded, while the attitude tracking errors fast converge to a sufficiently small region of the origin in finite time. The effectiveness and superiority of the designed control strategy are validated via contrastive simulation results. The highlights of the proposed control strategy are summarized as follows:
\begin{enumerate}[\hspace{0.5em}1)]
	\item To approximate the lumped disturbances of the small unmanned helicopter system, a new ASDO is proposed, which maintains the fast finite-time convergence and adaptability to disturbances of the adaptive second-order sliding mode observer (ASOSMO) \cite{ASTC2015}. Moreover, the designed ASDO has smoother output than ASOSMO, which leads to smoother control inputs for the helicopter system.
    \item A novel inequality is designed to handle the potential singularity problem in the design of finite-time backstepping controller. According to the inequality, a novel singularity-free virtual control law is constructed to achieve finite-time convergence while suppressing the potential singularity caused by the time derivative of the virtual control law. In addition, based on this inequality, an adaptive law with $\sigma $-modification term is designed to estimate the upper bound of the ASDO approximation error, thereby further improving the tracking performance.
	\item By integrating ASDO and FFTCF into the proposed control strategy, the attitude tracking errors of the helicopter system achieve finite-time convergence with faster response.
\end{enumerate}

The remainder of this paper is arranged as follows. Section 2 gives the problem formulation and some essential lemmas. Section 3 presents the design procedure of control laws. The contrastive simulation results are provided in Section 4. Some conclusions of this paper are drawn in Section 5.

Notations: In this paper, denote ${\mathop{\rm sig}\nolimits} {\left( x \right)^\gamma } = {\left| x \right|^\gamma }{\mathop{\rm sgn}} (x)$ and ${\mathop{\rm sgn}} ( \cdot )$ is the standard signum function. In addition, we use $\left\| \cdot \right\|$ to denote the Euclidean norm of vectors. ${\lambda _{\max }}\left(  \cdot  \right)$ and ${\lambda _{\min }}\left(  \cdot  \right)$ are used to represent the maximum and minimum eigenvalues of a matrix, respectively.
\section{Problem Formulation and Preliminaries}
\subsection{The small unmanned helicopter dynamics}
Because of similar dynamics with the real helicopter system, the 3-DOF lab helicopter system is frequently utilized as an ideal practical platform to validate various advanced control methods for helicopters \cite{10.Yang2020F}. The structure of 3-DOF helicopter system gives as Figure 1, which can mimic physical motions of elevation, pitch and travel movement. The helicopter is driven by front motor and back motor.  A positive voltage applied to each motor generates the positive elevation motion, and positive pitch motion is produced by giving a higher front motor voltage. By tilted thrust vectors when the helicopter body is pitching, the travel motion is generated. Moreover, an active disturbance system (ADS) serving as lumped disturbances is installed on the arm.
\begin{figure}[htb]\centering
	\includegraphics[width=0.5\textwidth]{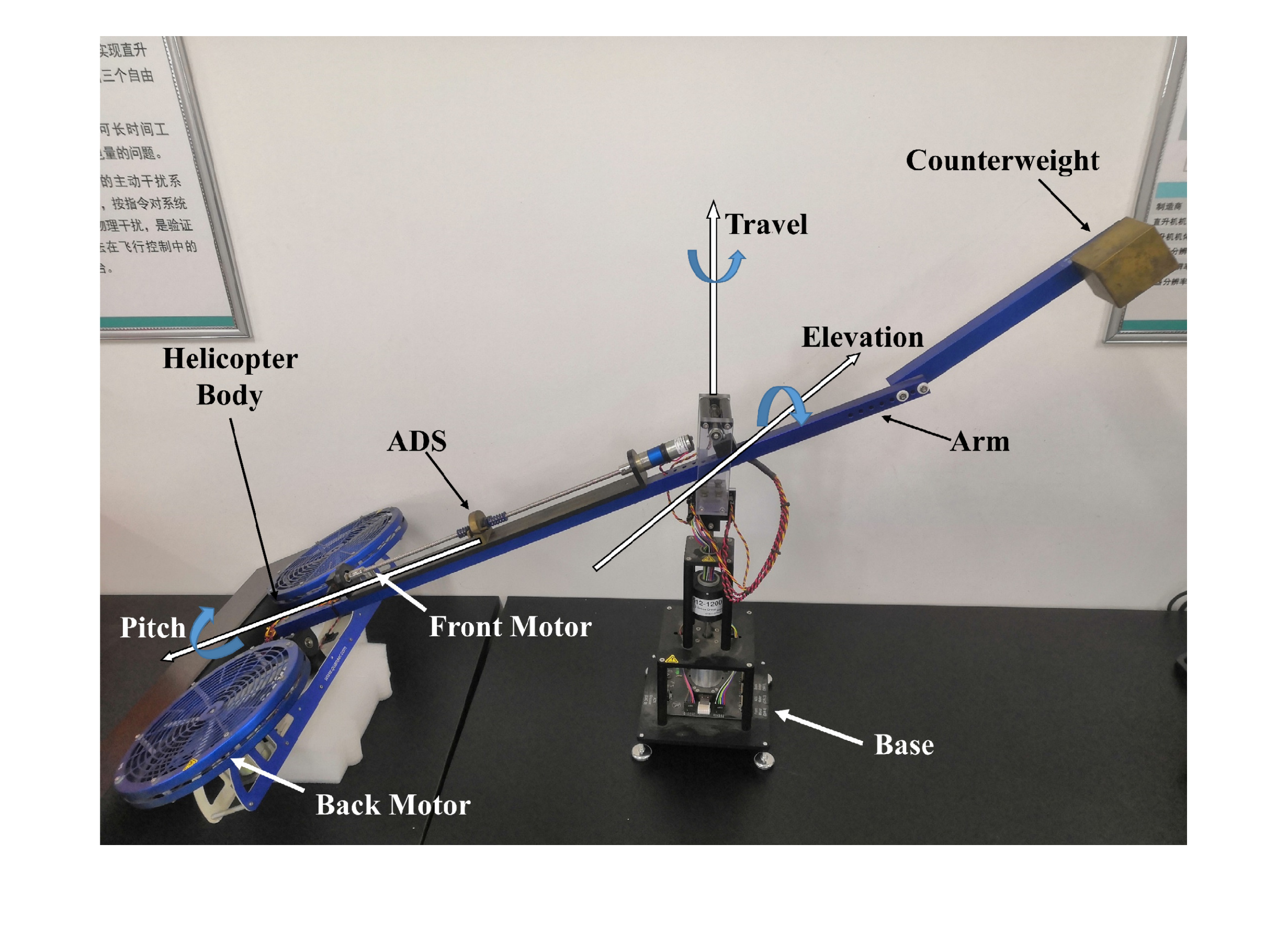}
	\caption{Structure of 3-DOF helicopter system}
\end{figure}

It is clear that the 3-DOF helicopter system is under-actuated, because only two of the 3-DOF can be controlled directly to move freely in the operating domain. In our work, we investigate the elevation and pitch motions, while considering the travel axis moves freely. Therefore, one can derive the models of elevation and pitch channels as follows \cite{2.Li2015F}
\begin{equation}
\begin{aligned}
{J_\alpha }\ddot \alpha &= {K_f}{L_a}\cos (\beta )({V_f} + {V_b}) - {m_e}g{L_a}\cos (\alpha ),\\
{J_\beta }\ddot \beta & = {K_f}{L_h}({V_f} - {V_b}),
\end{aligned}
\end{equation}
where $\alpha $ and $\beta $ are elevation angle and pitch angle, respectively. Definitions of other symbols can be found in Table \ref{table:1}.

\begin{table}[htb]
	\renewcommand{\arraystretch}{1.3}
	\caption{The parameters of the 3-DOF helicopter system}\label{table:1}
	\centering
	\resizebox{\columnwidth}{!}{
		\begin{tabular}{l l l}
			\hline\hline \\[-3mm]
			\multicolumn{1}{c}{Symbol} & \multicolumn{1}{c}{Definition} & \multicolumn{1}{l}{Value} \\ \hline
			${J_\alpha }$  & Moment of inertia (elevation axis) & $ 1.0348\, kg\cdot m^2 $ \\
            ${J_\beta }$  &  Moment of inertia (pitch axis) & $0.0451\, kg\cdot m^2 $ \\
			\pbox{20cm}{${L_a}$\\\hphantom{1}} & \pbox{20cm}{Distance between elevation axis and\\ 
             \hphantom{1}helicopter body center} &  \pbox{20cm}{$0.6600\, m $ \\ \hphantom{1}} \\ 
			${L_h}$ & Distance between pitch axis and either motor & $0.1780\,m $ \\
			${m_e}$ & Effective mass of the helicopter & $0.094\,kg $ \\
			${g}$ & Gravitational acceleration constant & $9.81\,m/s^2 $ \\
			${K_f}$ & Propeller force-thrust constant & $0.1188\,N/V $ \\
			${V_f}$ & Front motor voltage input & $[-24,24]\,V$ \\
			$ {V_b} $ & Back motor voltage input & $[-24,24]\,V $ \\
			\hline\hline
		\end{tabular}
	}
\end{table}
Denote ${x_1} = \alpha ,{x_2} = \dot \alpha ,{x_3} = \beta ,{x_4} = \dot \beta $. In view of external disturbance and system uncertainties, the model of helicopter system can be rewritten as
\begin{equation}\label{system1}
\begin{aligned}
{{\dot x}_1} &= {x_2},\\
{{\dot x}_2} &= \frac{{{L_a}}}{{{J_\alpha }}}\cos ({x_3}){u_1} - \frac{g}{{{J_\alpha }}}m_e{L_a}\cos ({x_1}) + {d_1}(\bar x,t),\\
{{\dot x}_3} &= {x_4},\\
{{\dot x}_4} &= \frac{{{L_h}}}{{{J_\beta }}}{u_2} + {d_2}(\bar x,t),
\end{aligned}
\end{equation}
where $\bar x = {[{x_1},{x_2},{x_3},{x_4}]^T}$ denotes the state vector of system \eqref{system1}, which is assumed to be measurable. Limited by the operating domain of the helicopter platform, $\cos \left( {{x_3}} \right) \in \left[ {{{\sqrt 2 } \mathord{\left/{\vphantom {{\sqrt 2 } 2}} \right.\kern-\nulldelimiterspace} 2},1} \right]$. ${d_1}\left(\bar x,t \right)$ and ${d_2}\left(\bar x,t \right)$ represent the lumped disturbances in the corresponding control channels. In addition, ${u_1}$ and ${u_2}$ are defined by
\begin{equation}
\begin{aligned}
{u_1} &= {K_f}({V_f} + {V_b}),\\
{u_2} &= {K_f}({V_f} - {V_b}).
\end{aligned}
\end{equation}. 

The control objective is to design ASDO-based control laws ${u_1},{u_2}$ such that the tracking errors of elevation and pitch angles fast converge to a sufficiently small region of the origin in finite time. 

\begin{assumption}
The reference trajectories ${x_{1d}}(t),{x_{3d}}(t)$ and their first derivatives are smooth, known, and bounded.
\end{assumption}

\begin{assumption}[\cite{2021assum1,2021assum2,2021assum3,2022assum1}]
The time derivatives of lumped disturbances are bounded, while the upper bounds are unknown, i.e., there exist unknown constants $\delta_{1}$, $\delta_{2}$ such that $\left| {{{\dot d}_1}\left(\bar x,t \right)} \right| \le {\delta_{1}}$, $\left| {{{\dot d}_2}\left(\bar x,t \right)} \right| \le {\delta_{2}}$, with ${\delta _1} \ge 0,{\delta _2} \ge 0$. 
\end{assumption}
\subsection{Lemmas}
To better describe the following \emph{Lemmas}, we consider a general system
\begin{equation}
\dot x = h(x(t)),{x_0} = x(0),
\end{equation}
where $h:{W_0} \to {\mathbb{R}}^n$  is continuous on an open neighborhood ${W_0} \subset {\mathbb{R}}^n$ of the origin and $h(0) = 0$.  $x(t,{x_0})$ denotes the solution of system (4), which is in the sense of Filippov \cite{1999Filippov}.

\begin{lemma} [\cite{2021fnt4,lemma1}]
Suppose there exists a continuous and positive-definite function $V\left( x \right)$ satisfying
\begin{equation}
\dot V(x) \le  - {\lambda _1}V(x) - {\lambda _2}V{(x)^\gamma }{\rm{ + }}{\lambda _3},
\end{equation}
where ${\lambda _1} > 0,{\lambda _2} > 0,0 \le {\lambda _3} < \infty ,\gamma  \in \left( {0,1} \right)$.

1) If ${\lambda _3}=0$, the origin of system (4) is fast finite-time stable.

2) If ${\lambda _3} \in \left( {0, \infty} \right)$, the trajectory of system (4) is fast finite-time uniformly ultimately boundedness, and the residual set is bounded by:
\begin{equation}
D_x = \left\{ {x:V\left( x \right) \le \min \left\{ {\frac{{{\lambda _3}}}{{\left( {1 - \theta } \right){\lambda _1}}},{{\left( {\frac{{{\lambda _3}}}{{\left( {1 - \theta } \right){\lambda _2}}}} \right)}^{\frac{1}{\gamma }}}} \right\}} \right\},
\end{equation}
where $\theta  \in (0,1)$.
\end{lemma}

\begin{lemma} [\cite{lemma2}]
Suppose there exists a continuous and positive-definite function $V:{W_0} \to R$, such that the following condition holds: 
\begin{equation}
\dot V(x) \le  - {c_1}V{(x)^{{\gamma_1}}} - {c_2}V(x){\rm{ + }}{c_3}V{(x)^{{\gamma_2}}},
\end{equation}
where ${c_1} > 0,{c_2} > 0,{c_3} > 0,{\gamma_1} \in (0,1),{\gamma_2} \in (0,{\gamma_1})$, then the trajectory of system (4) is fast finite-time uniformly ultimately boundedness. In addition, the convergent region is given by:
\begin{equation}
D_x = \left\{ {x:{\theta _1}V{{(x)}^{{\gamma_1} - {\gamma_2}}} + {\theta _2}V{{(x)}^{1 - {\gamma_2}}} < {c_3}} \right\},
\end{equation}
where ${\theta _1} \in (0,{c_1}),{\theta _2} \in (0,{c_2})$.
\end{lemma}
\begin{remark}
The meaning of ``fast"  in \emph{Lemma 1} and \emph{Lemma 2} is that the solution of system (4) can quickly converge to the origin or its neighborhood regardless of the distance between initial state and origin, while the convergence rate of original finite-time stability becomes slow when initial state is far from the origin.
\end{remark}

To approximate the virtual control signal and its derivative, the following FFTCF is introduced \cite{lemma3}
\begin{equation}
\begin{aligned}
{{\dot x}_{1,c}} = & {x_{2,c}},\\
{\varepsilon_c ^2}{{\dot x}_{2,c}} = & - {a_0}\left( {{x_{1,c}} - {\alpha _r}\left( t \right)} \right) - {a_1}{\mathop{\rm sig}\nolimits} {\left( {{x_{1,c}} - {\alpha _r}\left( t \right)} \right)^{{\gamma _3}}} \\
&- {b_0}\varepsilon_c {x_{2,c}} - {b_1}{\mathop{\rm sig}\nolimits} {\left( {\varepsilon_c {x_{2,c}}} \right)^{{\gamma _4}}},
\end{aligned}
\end{equation}
where ${\alpha _r}\left( t \right)$ is the input signal. $\varepsilon_c > 0$ is a perturbation parameter, and ${a_0},{a_1},{b_0},{b_1},{\gamma _3},{\gamma _4}$ are appropriately tuning parameters satisfying ${a_0} > 0,{a_1} > 0,{b_0} > 0,{b_1} > 0,{\gamma _4} \in \left( {0,1} \right),{\gamma _3} \in \left( {{\gamma _4}/\left( {2 - {\gamma _4}} \right),1} \right)$. The following \emph{Lemma} holds.

\begin{lemma} [\cite{lemma3}]
Suppose that ${\alpha _r}\left( t \right)$ is a continuous and piecewise twice differentiable signal. For differentiator (18), there exist $\rho  > 0\left( {\rho {\gamma _4} > 2} \right)$ and $\Gamma  > 0$ such that
\begin{equation}
\begin{aligned}
{x_{1,c}} - {\alpha _r}\left( t \right) = {\rm O}\left( {{\varepsilon_c ^{\rho {\gamma _4}}}} \right),{x_{2,c}} - {\dot \alpha _r}\left( t \right) = {\rm O}\left( {{\varepsilon_c ^{\rho {\gamma _4} - 1}}} \right)\left( {t > \varepsilon_c \Gamma } \right),
\end{aligned}
\end{equation}
where ${\rm O}\left( {{\varepsilon_c ^{\rho {\gamma _4}}}} \right)$ denotes that the approximation error between $x_{1,c}$ and ${\alpha _r}\left( t \right)$ is ${\varepsilon_c ^{\rho {\gamma _4}}}$ order.
\end{lemma}

\begin{lemma} [\cite{lemma4}]
For arbitrary ${\chi _1},{\chi _2} \in {\mathbb{R}}$, if constants ${\tau _1} > 0,{\tau _2} > 0$, the following inequality holds:
\begin{equation}
{\left| {{\chi _1}} \right|^{{\tau _1}}}{\left| {{\chi _2}} \right|^{{\tau _2}}} \le \frac{{{\tau _1}}}{{{\tau _1}{\rm{ + }}{\tau _2}}}{\left| {{\chi _1}} \right|^{{\tau _1}{\rm{ + }}{\tau _2}}}{\rm{ + }}\frac{{{\tau _2}}}{{{\tau _1}{\rm{ + }}{\tau _2}}}{\left| {{\chi _2}} \right|^{{\tau _1}{\rm{ + }}{\tau _2}}}.
\end{equation}
\end{lemma}

\begin{lemma} [\cite{lemma5}]
For ${y_j} \in {\mathbb{R}},j = 1,2, \cdots ,n,0 < h \le 1$, it holds that
\begin{equation}
{\left( {\sum\limits_{j = 1}^n {\left| {{y_j}} \right|} } \right)^h} \le \sum\limits_{j = 1}^n {{{\left| {{y_j}} \right|}^h} \le {n^{1 - h}}} {\left( {\sum\limits_{j = 1}^n {\left| {{y_j}} \right|} } \right)^h}.
\end{equation}
\end{lemma}

\begin{lemma} [\cite{lemma6}]
For $x,y \in {\mathbb{R}}$, if $r = {r_2}/{r_1} < 1$ with ${r_1},{r_2}$ being positive odd integers, then $x{\left( {y - x} \right)^r} \le  - {h_1}{x^{1 + r}} + {h_2}{y^{1 + r}}$, where
\begin{equation}
\begin{aligned}
&{h_1} = \frac{1}{{1 + r}}\left[ {{2^{r - 1}} - {2^{\left( {r - 1} \right)\left( {r + 1} \right)}}} \right],\\
&{h_2} = \frac{1}{{1 + r}}\left[ {\frac{{2r + 1}}{{r + 1}} + \frac{{{2^{ - {{\left( {r - 1} \right)}^2}\left( {r + 1} \right)}}}}{{r + 1}} - {2^{r - 1}}} \right].
\end{aligned}
\end{equation}
\end{lemma}

To address the potential singularity problem in the design of virtual control law, we present a novel inequality as follows:
\begin{lemma} 
For any variable $x \in {\mathbb{R}}$ and any constants $\varepsilon  > 0, \sigma >0$, one has
\begin{equation}
0 \le \left| x \right| - {x^2}\sqrt {\frac{{{x^2} + {\sigma^2}+{\varepsilon ^2}}}{{\left( {{x^2} + {\varepsilon ^2}} \right)\left( {{x^2} + {\sigma^2} } \right)}}}  < \frac{{\varepsilon \sigma }}{{\sqrt {{\varepsilon ^2} + {\sigma ^2}} }}.  
\end{equation}
\end{lemma}

It is easy to verify the correctness of \emph{Lemma 7}, and therefore the proof is omitted. 

\section{Main Results}
In this section, we explain the control law design of the elevation channel in detail, while the control strategy of the pitch channel can be developed in a similar process.
\subsection{ASDO Design}
For the elevation channel
\begin{equation}
\begin{aligned}
&{{\dot x}_1} = {x_2},\\
&{{\dot x}_2} = \frac{{{L_a}}}{{{J_\alpha }}}\cos ({x_3}){u_1} - \frac{g}{{{J_\alpha }}}m_e{L_a}\cos ({x_1}) + {d_1}(\bar x,t),
\end{aligned}
\end{equation}
the ASDO is designed as
\begin{equation}
\begin{aligned}
&{{\hat d}_1}{\rm{ = }}{L_{1d}}(t){\left| {{s_d}} \right|^{\frac{{m - 1}}{m}}}{\mathop{\rm sgn}} ({s_d}) + {L_{2d}}(t){s_d} + {\varphi _d},\\
&{{\dot \varphi }_d} = {L_{3d}}(t){\left| {{s_d}} \right|^{\frac{{m - 2}}{m}}}{\mathop{\rm sgn}} ({s_d}) + {L_{4d}}(t){s_d},
\end{aligned}
\end{equation}
where ${\hat d}_1$ represents the approximation of $d_1$, and
\begin{equation}
\begin{aligned}
&{s_d} = {x_2} - {{\hat x}_2},\\
&{{\dot {\hat x}}_2} = \frac{{{L_a}}}{{{J_\alpha }}}\cos ({x_3}){u_1} - \frac{g}{{{J_\alpha }}}m_e{L_a}\cos ({x_1}) + {{\hat d}_1}.
\end{aligned}
\end{equation}

The adaptive gains ${L_{1d}}(t),{L_{2d}}(t),{L_{3d}}(t),{L_{4d}}(t)$ are expressed as
\begin{equation}
\begin{aligned}
{L_{1d}}(t)&={k_{1}}L_d^{\textstyle{{m - 1} \over m}}\left( t \right),{\rm{  }}{L_{2d}}(t) = {k_{2}}L_d \left( t \right),\\
{L_{3d}}(t)& = {k_{3}}L_d^{{\textstyle{{2m - 2} \over m}}}\left( t \right),{\rm{  }}{L_{4d}}(t) = {k_{4}}L_d^2 \left( t \right).
\end{aligned}
\end{equation}
where ${k_{1}},{k_{2}},{k_{3}},{k_{4}},m$ are positive constants satisfying
\begin{equation}
{m^2}{k_{3}}{k_{4}} > \left( {\frac{{{m^3}{k_{3}}}}{{m - 1}} + \left( {4{m^2} - 4m + 1} \right)k_{1}^2} \right)k_{2}^2,m > 2.
\end{equation}

${L_d}(t)$ is a scalar, positive, and time-varying function. The ${L_d}(t)$ satisfies
\begin{equation}
\dot L_d\left( t \right) =
\begin{cases}
{\kappa},&if \quad \left| s_d \right| \ge \varepsilon_d, \hfill \\
0,&else, \hfill \\
\end{cases} 
\end{equation}
where $\kappa$ is a positive constant, $\varepsilon_d$ is an arbitrary small positive value.
\begin{theorem}
Consider the elevation channel (15) under the condition of \emph{Assumption 2}, and the ASDO (16) with adaptive law (18), then one can obtain the following conclusions:

1) If $\dot d_1 = 0$, the approximation error of ASDO fast converges to the origin in finite time.

2) If $\dot d_1 \ne 0$, the approximation error of ASDO fast converges to a neighborhood of the origin in finite time.
\end{theorem}
\begin{proof}
Considering (15) and (17), we obtain
\begin{equation}
{\dot s_d} = {\dot x_2} - {{\dot {\hat x}}_2} = {d_1} - {\hat d_1}.
\end{equation}

Substituting (16) into (21) leads to
\begin{equation}
\begin{aligned}
&{{\dot s}_d} =  - {L_{1d}}(t){\left| {{s_d}} \right|^{\frac{{m - 1}}{m}}}{\mathop{\rm sgn}} ({s_d}) - {L_{2d}}(t){s_d} + {\varphi _{d1}},\\
&{{\dot \varphi }_{d1}} =  - {L_{3d}}(t){\left| {{s_d}} \right|^{\frac{{m - 2}}{m}}}{\mathop{\rm sgn}} ({s_d}) - {L_{4d}}(t){s_d} + {{\dot d}_1}.
\end{aligned}
\end{equation}
 
Define a new state vector
\begin{equation}
\xi_d  = \left[ {\begin{array}{*{20}{c}}
{{\xi _{1d}}}\\
{{\xi _{2d}}}\\
{{\xi _{3d}}}
\end{array}} \right] = \left[ {\begin{array}{*{20}{c}}
{L_d^{{\textstyle{{m - 1} \over m}}}\left( t \right){{\left| s_d \right|}^{{\textstyle{{m - 1} \over m}}}}{\mathop{\rm sgn}} (s_d)}\\
{L_d \left( t \right)s_d}\\
\varphi_{d1} 
\end{array}} \right].
\end{equation} 

After taking the derivative of $\xi_d$, we obtain
\begin{equation}
\begin{aligned}
\dot \xi_d  =& L_d^{{\textstyle{{m - 1} \over m}}}{\left| s_d \right|^{\frac{{ - 1}}{m}}}\left[ {\begin{array}{*{20}{c}}
{ - \frac{{m - 1}}{m}{k_{1}}}&{ - \frac{{m - 1}}{m}{k_{2}}}&{\frac{{m - 1}}{m}}\\
0&0&0\\
{ - {k_{3}}}&0&0
\end{array}} \right]\xi_d \\
 &+ L_d \left[ {\begin{array}{*{20}{c}}
0&0&0\\
{{\rm{ - }}{k_{1}}}&{ - {k_{2}}}&1\\
0&{ - {k_{4}}}&0
\end{array}} \right]\xi_d  + \left[ {\begin{array}{*{20}{c}}
{\frac{{m - 1}}{m}\frac{{{{\dot L}_d}(t)}}{{{L_d}(t)}}{\xi _{1d}}}\\
{\frac{{{{\dot L}_d}(t)}}{{{L_d}(t)}}{\xi _{2d}}}\\
{\dot d_1}
\end{array}} \right].
\end{aligned}
\end{equation}

Then a candidate Lyapunov function is chosen as $V\left( \xi_d  \right) = {\xi_d ^T}P\xi_d$, with
\begin{equation}
\begin{aligned}
P = \frac{1}{2}\left[ {\begin{array}{*{20}{c}}
{\frac{{2m}}{{m - 1}}{k_{3}} + k_{1}^2}&{{k_{1}}{k_{2}}}&{ - {k_{1}}}\\
{{k_{1}}{k_{2}}}&{2{k_{4}} + k_{2}^2}&{ - {k_{2}}}\\
{ - {k_{1}}}&{ - {k_{2}}}&2
\end{array}} \right],
\end{aligned}
\end{equation}
where $P$ is a symmetric positive definite matrix. Taking the derivative of $V\left( \xi_d  \right)$ along the trajectories of (22) yields 
\begin{equation}
\begin{aligned}
\dot V\left( \xi_d  \right) = & - {L_d}{\left| {{\xi _{1d}}} \right|^{{\textstyle{{ - 1} \over {m - 1}}}}}{\xi_d ^T}{\Omega _1}\xi_d  - L_d {\xi_d ^T}{\Omega _2}\xi_d  \\
&+ {\sigma _1}{\dot d_1}\xi_d  + \frac{{2m - 2}}{m}\frac{{{{\dot L}_d}(t)}}{{{L_d}(t)}}{\sigma _2}P\xi_d, 
\end{aligned}	
\end{equation}
where $\sigma_1 = \left[-k_{1} -k_{2} \quad 2 \right]$, $\sigma_2 =\left[\xi _{1d} \quad \frac{m}{m-1}{\xi _{2d}} \quad 0\right]$, and
\begin{equation}\small
\begin{aligned}
&{\Omega _1}= \frac{{{k_{1}}}}{m}\left[ {\begin{array}{*{20}{c}}
{{k_{3}}m + k_{1}^2\left( {m - 1} \right)}&0&{ - {k_{1}}\left( {m - 1} \right)}\\
0&{{k_{4}}m + k_{2}^2\left( {3m - 1} \right)}&{ - {k_{2}}\left( {2m - 1} \right)}\\
{ - {k_{1}}\left( {m - 1} \right)}&{ - {k_{2}}\left( {2m - 1} \right)}&{m - 1}
\end{array}} \right],\\
&{\Omega _2} = {k_{2}}\left[ {\begin{array}{*{20}{c}}
{{k_{3}} + k_{1}^2\left( {3m - 2} \right)/m}&0&0\\
0&{{k_{4}} + k_{2}^2}&{ - {k_{2}}}\\
0&{ - {k_{2}}}&1
\end{array}} \right].
\end{aligned}
\end{equation}

It is easy to prove that the matrices $\Omega _1$ and $\Omega _2$ both are positive definite with (19). By using
\begin{equation}
{\lambda _{\min }}\left( P \right){\left\| \xi_d  \right\|^2} \le V \le {\lambda _{\max }}\left( P \right){\left\| \xi_d  \right\|^2},
\end{equation}
(26) is formulated as
\begin{equation}
\begin{aligned}
\dot V \le & - {L_d}\left( t \right)\frac{{{\lambda _{\min }}\left( {{\Omega _1}} \right)}}{{\lambda _{\max }^{{\gamma_1}}\left( P \right)}}{V^{{\gamma_1}}} - L_d \left( t \right)\frac{{{\lambda _{\min }}\left( {{\Omega _2}} \right)}}{{{\lambda _{\max }}\left( P \right)}}V \\
&+ \frac{{\delta_1 {{\left\| {{\sigma _1}} \right\|}}}}{{\lambda _{\min }^{{\raise0.5ex\hbox{$\scriptstyle 1$}
\kern-0.1em/\kern-0.15em
\lower0.25ex\hbox{$\scriptstyle 2$}}}\left( P \right)}}{V^{\frac{1}{2}}} + \frac{{m - 1}}{m}\frac{{{{\dot L}_d}}}{{{L_d}}}{\xi_d ^T}Q\xi_d, 
\end{aligned}
\end{equation}
where ${\gamma_1} = \left( {2m - 3} \right)/\left( {2m - 2} \right)$, $Q = diag\left[ {{q_1},{q_2},{q_3}} \right]$ is a diagonal matrix with positive diagonal elements, which are expressed as follows
\begin{equation}
\begin{aligned}
{q_1}& = \frac{{2m}}{{m - 1}}{k_{3}} + k_{1}^2 + \frac{{\left( {2m - 1} \right){k_{1}}{k_{2}}}}{{2\left( {m - 1} \right)}} + \frac{{{k_{1}}}}{2},\\
{q_2} &= \frac{m}{{2\left( {m - 1} \right)}}\left( {4{k_{4}} + 2k_{2}^2 + {k_{2}}} \right) + \frac{{\left( {2m - 1} \right){k_{1}}{k_{2}}}}{{2\left( {m - 1} \right)}},\\
{q_3} &= \frac{{{k_{1}}}}{2} + \frac{{m{k_{2}}}}{{2\left( {m - 1} \right)}}.
\end{aligned}
\end{equation}

Then (29) is further rewritten as
\begin{equation}\small
\dot V \le  - {L_d}\left( t \right){n_1}{V^{{\gamma_1}}} + {n_2}{V^{\frac{1}{2}}} - \left( {L_d \left(t\right){n_3} - \frac{{2m - 2}}{m}{n_4}\frac{{{{\dot L}_d}}}{{{L_d}}}} \right)V,
\end{equation}
where
\begin{equation}
\begin{aligned}
{n_1} = \frac{{{\lambda _{\min }}\left( {{\Omega _1}} \right)}}{{\lambda _{\max }^{{\gamma_1}}\left( P \right)}},
{n_2} = \frac{{\delta_1 {{\left\| {{\sigma _1}} \right\|}}}}{{\lambda _{\min }^{{\raise0.5ex\hbox{$\scriptstyle 1$}
\kern-0.1em/\kern-0.15em
\lower0.25ex\hbox{$\scriptstyle 2$}}}\left( P \right)}},
{n_3} = \frac{{{\lambda _{\min }}\left( {{\Omega _2}} \right)}}{{{\lambda _{\max }}\left( P \right)}},
{n_4} = \frac{{{\lambda _{\max }}\left( Q \right)}}{{2{\lambda _{\min }}\left( P \right)}}.
\end{aligned}
\end{equation}

1) If $\dot d_1= 0$, (31) becomes
\begin{equation}
\dot V \le  - {L_d}\left( t \right){n_1}{V^{{\gamma_1}}} - \left( {L_d \left( t \right){n_3} - \frac{{2m - 2}}{m}{n_4}\frac{{{{\dot L}_d}}}{{{L_d}}}} \right)V.
\end{equation}

Due to ${\dot L_d}\left( t \right) \ge 0$, $L_d \left( t \right){n_3} - \left( {2m - 2} \right){n_4}{\dot L_d}/\left( {{L_d}m} \right)$ is positive in finite time. It follows
from (33) that 
\begin{equation}
\dot V \le  - {c_1}{V^{{\gamma_1}}} - {c_2}V,
\end{equation}
where ${c_1}$ and ${c_2}$ are positive constants, ${\gamma_1} \in (0.5,1)$. Based on \emph{Lemma 1}, $\xi_d$ fast converges to origin in  finite time, then $s_d,\dot s_d$ fast converge to the origin in finite time. According to (21), the approximation error of ASDO can fast converge to the origin in finite time. 

2) If $\dot d_1 \ne 0$, with the same analysis of 1), it follows from (31) that
\begin{equation}
\dot V \le  - {c_4}{V^{{\gamma_1}}} - {c_5}V{\rm{ + }}{c_3}{V^{{\textstyle{1 \over 2}}}},
\end{equation}
where ${c_3},{c_4}$ and ${c_5}$ are positive constants, ${\gamma_1} \in (0.5,1)$. Based on \emph{Lemma 2}, $\xi_d$ fast converges to a neighborhood of the origin in  finite time. In addition, the convergent region is given by
\begin{equation}
D = \left\{ {\xi_d :{\theta _1}V{{(\xi_d )}^{{\gamma_1} - {\gamma_2}}} + {\theta _2}V{{(\xi_d )}^{1 - {\gamma_2}}} < {c_3}} \right\},
\end{equation}
where ${\theta _1} \in (0,{c_4}),{\theta _2} \in (0,{c_5})$, ${\gamma_2} = 0.5$.

Define an auxiliary variable ${\theta _3} \in \left( {0,1} \right)$. If ${\theta _3}$ is selected satisfying
\begin{equation}
{\theta _3}^{1 - {\gamma_2}}{\theta _2}^{{\gamma_1} - {\gamma_2}}{c_3}^{1 - {\gamma_2}} = {\theta _1}^{1 - {\gamma_2}}{(1 - {\theta _3})^{{\gamma_1} - {\gamma_2}}},
\end{equation}
$\xi_d$ can converge to $D = {D_1} = {D_2}$ in finite time, where
\begin{equation}
\begin{aligned}
{D_1} &= \left\{ {\xi_d :V{{(\xi_d )}^{{\gamma_1} - {\gamma_2}}} < {\theta _3}{c_3}/{\theta _1}} \right\},\\
{D_2} &= \left\{ {\xi_d :V{{(\xi_d )}^{1 - {\gamma_2}}} < \left( {1 - {\theta _3}} \right){c_3}/{\theta _2}} \right\}.
\end{aligned}
\end{equation}

Design a region $D_4$ as
\begin{equation}
\begin{aligned}
{D_4} &= \left\{ {\xi_d :{\lambda _{\min }}\left( P \right){{\left\| \xi_d  \right\|}^2} < {{\left( {1 - {\theta _3}} \right)}^2}c_3^2/{\theta _2}^2} \right\}\\
 &= \left\{ {\xi_d :\left\| \xi_d  \right\| < {\Delta}} \right\},
\end{aligned}
\end{equation}
where ${\Delta}{\rm{ = }}{\lambda _{\min }}{\left( P \right)^{ - 1/2}}\left( {1 - {\theta _3}} \right){c_3}/{\theta _2}$. 

In terms of (28), it follows from (38) and (39) that $D_4$ contains $D_2$. Considering the definition of $\xi_d$, the following inequalities $\left\| {{\xi _{1d}}} \right\| \le \left\| \xi_d  \right\|$, $\left\| {{\xi _{2d}}} \right\| \le \left\| \xi_d  \right\|$ and $\left\| {{\xi _{3d}}} \right\| \le \left\| \xi_d  \right\|$ hold.

Then the set ${D_5} = \left\{ {{\xi _{1d}},{\xi _{2d}},{\xi _{3d}}:\left\| {{\xi _{1d}}} \right\| < {\Delta},\left\| {{\xi _{2d}}} \right\| < {\Delta},\left\| {{\xi _{3d}}} \right\| < {\Delta}} \right\}$ contains the set $D_4$. Therefore, since $\xi_d$ converges to $D_1$ in $T_4$, it will also converge to $D_5$ in finite time. Using (18), (22) and (23), we can obtain that $s_d,\dot s_d$ fast converge to a neighborhood of the origin in finite time. According to (21), the approximation error of ASDO fast converges to a neighborhood of the origin in finite time. 
\end{proof}
\begin{remark}
By selecting different values of $m$ $(m>2)$ in (16), a series of ASDOs are obtained. If $m = 2$, (16) is transformed into the ASOSMO proposed in \cite{ASTC2015}. Moreover, the designed ASDO maintains the fast finite-time convergence and adaptability to disturbance of ASOSMO while having smoother output than ASOSMO, which results in smoother control input. The superiority of ASDO will be validated in the next section. 
\end{remark}
\subsection{Controller Design}
Define the following error variables for the elevation channel:
\begin{equation}
\begin{aligned}
{z_1} = {x_1} - {x_{1d}},
{z_2} = {x_2} - {x_{1,c}},
\end{aligned}
\end{equation}
where ${x_{1,c}}$ is the estimation of the virtual control signal ${\alpha _r}$ via FFTCF.

In order to compensate the error caused by the FFTCF estimation, the following auxiliary dynamic system is employed \cite{2021fnt1}:
\begin{equation}
\begin{aligned}
&{{\dot \xi }_1} =  - {\bar k_1}{\xi _1} + {\xi _2} + \left( {{x_{1,c}} - {\alpha _r}} \right) - {l_1}\xi _1^r,\\
&{{\dot \xi }_2} =  - {\bar k_2}{\xi _2} - {\xi _1} - {l_2}\xi _2^r,
\end{aligned}
\end{equation}
where ${\xi _1},{\xi _2}$ are the error compensation signals with ${\xi _1}\left( 0 \right) = 0,{\xi _2}\left( 0 \right) = 0$. ${\bar k_1} > 0,{\bar k_2} > 0,{l_1} > 0,{l_2} > 0$ are the appropriately tuning parameters and $r = {r_2}/{r_1} < 1$ with ${r_1} > 0,{r_2} > 0$ being odd integers.

Denote ${v_1},{v_2}$ as the compensated tracking errors, which are formulated as follows
\begin{equation}
\begin{aligned}
{v_1} = {z_1} - {\xi _1},
{v_2} = {z_2} - {\xi _2}.
\end{aligned}
\end{equation}

The singularity-free virtual control signal ${\alpha _r}$ is developed as
\begin{equation}
\begin{aligned}
{\alpha _r} =  - {\bar k_1}{z_1} + {{\dot x}_{1d}} - {s_1}v_1^{1 + 2r}\sqrt {\frac{{v_1^{2 + 2r} + {\sigma _r^2} + \varepsilon _r^2}}{{\left( {v_1^{2 + 2r} + \varepsilon _r^2} \right)\left( {v_1^{2 + 2r} + {\sigma _r^2}} \right)}}},
\end{aligned}
\end{equation}
where ${s_1} > 0,{\varepsilon _r} > 0, {\sigma _r} >0$ are the appropriately tuning parameters.

Then, the controller ${u_1}$ is constructed as
\begin{equation}
\begin{aligned}
{{u}_1} = \frac{{{J_\alpha }}}{{{L_a}\cos \left( {{x_3}} \right)}}&\left( { - {\bar k_2}{z_2} - {z_1} + {\dot x_{1,c}} + \frac{g}{{{J_\alpha }}}m_e{L_a}\cos ({x_1})}\right.\\
 &\phantom{=\;\;}\left.{ - {s_2}v_2^r - {{\hat d}_1} - {\hat p{v_2}\sqrt {\frac{{v_2^2 + {\sigma _p^2} + \varepsilon _p^2}}{{\left( {v_2^2 + \varepsilon _p^2} \right)\left( {v_2^2 + {\sigma _p^2}} \right)}}} }} \right),
\end{aligned}
\end{equation}
where ${s_2} > 0, {\varepsilon _p} > 0, {\sigma _p} >0$ are the appropriately tuning parameters, and $\hat p$ is the approximation of ${d^ * }$ with ${d^ * }$ being the upper bound of ASDO approximation error.

To further attenuate the observer approximation error, the adaptive law with $\sigma $-modification term $\hat p$ is designed as follows
\begin{equation}
\dot {\hat p} = q\left[ {v_2^2\sqrt {\frac{{v_2^2 + {\sigma _p^2} + \varepsilon _p^2}}{{\left( {v_2^2 + \varepsilon _p^2} \right)\left( {v_2^2 + {\sigma _p^2}} \right)}}}  - \mu \hat p - \eta {{\hat p}^r}} \right],
\end{equation}
where $q > 0,\mu  > 0,\eta  > 0$ are the appropriately tuning parameters.
\subsection{Stability Analysis}
\begin{theorem}
Consider the elevation channel system (15) satisfying \emph{Assumptions 1 and 2}. If the ASDO is presented as (16), the FFTCF is selected as (9), the auxiliary dynamic system is established as (41), the virtual control signal is developed as (43), and the adaptive law with $\sigma $-modification term is designed as (45), then we can construct the control law ${{u}_1}$ as (44) such that all the closed-loop system signals are fast finite-time bounded, while the attitude tracking error fast converges to a small neighborhood of the origin in finite time.
\end{theorem}
\begin{proof}
By employing \emph{Theorem 1}, the following conclusion can be drawn: there exists a positive constant ${d^ * }$, such that $\left| {{{\tilde d}_1}} \right| \le {d^ * }$ for all $t \ge {t_1}$, where ${\tilde d_1} = {d_1} - {\hat d_1}$ denotes the ASDO approximation error and ${t_1}$ denotes the convergent time. In addition, based on \emph{Lemma 3}, it is obtained that $\left| {x_{1,c}} - {\alpha _r}\left( t \right) \right| = {\rm O}\left( {{\varepsilon_c ^{\rho {\gamma _4}}}} \right)$ is achieved in finite time $t_2$.

The Lyapunov function is selected as
\begin{equation}
V = \frac{1}{2}{q^{ - 1}}{\left( {\hat p - {d^ * }} \right)^2} + \sum\limits_{i = 1}^2 {\frac{1}{2}\left( {v_i^2 + \xi _i^2} \right)}.
\end{equation}

According to \emph{Lemma 4}, the following inequalities holds:
\begin{equation}
\begin{aligned}
{l_i}{v_i}\xi _i^r \le \frac{{{l_i}}}{{1 + r}}v_i^{r + 1} + \frac{{{l_i}r}}{{1 + r}}\xi _i^{r + 1},i = 1,2.
\end{aligned}
\end{equation}

Taking the time derivative of $V$ and substituting (14), (15), (40), (41), (42), (43), (44), (45) and (47) into it, when $t > {t_1}$, we obtain
\begin{equation}
\begin{aligned}
\dot V \le & - \left[ {\sum\limits_{i = 1}^2 {{{\bar k}_i}v_i^2}  + \sum\limits_{i = 1}^2 {{{\bar k}_i}\xi _i^2} } \right] + {\xi _1}\left( {{x_{1,c}} - {\alpha _r}} \right) - \mu (\hat p - {d^ * })\hat p \\
&- \left[ {\sum\limits_{i = 1}^2 {\left( {{s_i} - \frac{{{l_i}}}{{1 + r}}} \right)v_i^{1 + r} + \sum\limits_{i = 1}^2 {\frac{{{l_i}}}{{1 + r}}\xi _i^{1 + r}} } } \right] - \eta (\hat p - {d^ * }){\hat p^r} \\
&+ \frac{{{\varepsilon _p}{\sigma _p}}}{{\sqrt {\varepsilon _p^2 + \sigma _p^2} }}{d^ * } + \frac{{{\varepsilon _r}{\sigma _r}}}{{\sqrt {\varepsilon _r^2 + \sigma _r^2} }}{s_1}.
\end{aligned}
\end{equation}

When $t > {t_2}$, one has
\begin{equation}
{\xi _1}\left( {{x_{1,c}} - {\alpha _r}} \right) \le \frac{1}{2}\xi _1^2 + \frac{1}{2}{\rm O}\left( {\varepsilon _c^{2\rho {\gamma _4}}} \right).
\end{equation}

Based on \emph{Lemma 6}, we get
\begin{equation}
\begin{aligned}
&- \mu  (\hat p - {d^ * }){{\hat p}} \le  - \frac{\mu }{2} {(\hat p - {d^ * })^{2}} + \frac{\mu }{2} {d^ * }^{2},\\
&- \eta (\hat p - {d^ * }){{\hat p}^r} \le  - \eta {h_1}{(\hat p - {d^ * })^{1 + r}} + \eta {h_2}{d^ * }^{1 + r}.
\end{aligned}
\end{equation}

Substituting (49) and (50) into (48), when $t > \max \left\{ {{t_1},{t_2}} \right\}$, we have
\begin{equation}
\begin{aligned}
\dot V \le  &- \left[ {\left( {{\bar k_1} - \frac{1}{2}} \right){\xi _1}^2 + {\bar k_2}{\xi _2}^2 + \frac{\mu }{2}{{\left( {\hat p - {d^ * }} \right)}^2} + \sum\limits_{i = 1}^2 {{\bar k_i}v_i^2} } \right] \\
&- \left[ {\sum\limits_{i = 1}^2 {\left( {{s_i} - \frac{{{l_i}}}{{1 + r}}} \right)v_i^{1 + r} + \sum\limits_{i = 1}^2 {\frac{{{l_i}}}{{1 + r}}\xi _i^{1 + r}} } } \right] - {h_1}\eta {\left( {\hat p - {d^ * }} \right)^{1 + r}}\\ 
&+ {\lambda _3}, 
\end{aligned}
\end{equation}
where ${\lambda _3} = {h_2}\eta {d^ * }^{1 + r} + 0.5\mu {d^ * }^2 + 0.5{\rm O}\left( {{\varepsilon_c ^{2\rho {\gamma _4}}}} \right)+\frac{{{\varepsilon _p}{\sigma _p}}}{{\sqrt {\varepsilon _p^2 + \sigma _p^2} }}{d^ * } + \frac{{{\varepsilon _r}{\sigma _r}}}{{\sqrt {\varepsilon _r^2 + \sigma _r^2} }}{s_1} > 0$. 

By utilizing \emph{Lemma 5}, (51) is further rewritten as
\begin{equation}
\dot V(x) \le  - {\lambda _1}V(x) - {\lambda _2}V{(x)^{\frac{{1 + r}}{2}}}{\rm{ + }}{\lambda _3},
\end{equation}
where
\begin{equation}
\begin{aligned}
&{\lambda _1} = \min \left\{ {\left( {2{\bar k_1} - 1} \right),2{\bar k_2},q\mu } \right\},\\
&{\lambda _2} = \min \left\{ {\left( {{s_i} - \frac{{{l_i}}}{{1 + r}}} \right){2^{\frac{{1 + r}}{2}}},\frac{{{l_i}}}{{1 + r}}{2^{\frac{{1 + r}}{2}}},{h_1}\eta {{\left( {2q} \right)}^{\frac{{1 + r}}{2}}}} \right\},
\end{aligned}
\end{equation}
with $i=1,2$.

Then selecting appropriate parameters, and in terms of \emph{Lemma 1}, ${v_i},{\xi _i}$ fast converge to the following region:
\begin{equation}
\begin{aligned}
&\left| {{v_i}} \right| \le \min \left\{ {\sqrt {\frac{{2{\lambda _3}}}{{\left( {1 - \theta } \right){\lambda _1}}}} ,\sqrt {2{{\left( {\frac{{{\lambda _3}}}{{\left( {1 - \theta } \right){\lambda _2}}}} \right)}^{\frac{2}{{1 + r}}}}} } \right\},\\
&\left| {{\xi _i}} \right| \le \min \left\{ {\sqrt {\frac{{2{\lambda _3}}}{{\left( {1 - \theta } \right){\lambda _1}}}} ,\sqrt {2{{\left( {\frac{{{\lambda _3}}}{{\left( {1 - \theta } \right){\lambda _2}}}} \right)}^{\frac{2}{{1 + r}}}}} } \right\},
\end{aligned}
\end{equation}
within finite time ${T}$, which is bounded by
\begin{equation}
\begin{aligned}
{T} \le \max \left\{ {{t_1},{t_2}} \right\} + &\max \left\{ {\frac{2}{{\theta {\lambda _1}\left( {1 - r} \right)}}\ln \frac{{\theta {\lambda _1}{V^{\frac{1-r}{2}}}\left( {{0}} \right) + {\lambda _2}}}{{{\lambda _2}}},}\right.\\
 &\phantom{=\;\;}\left.{\frac{2}{{{\lambda _1}\left( {1 - r} \right)}}\ln \frac{{{\lambda _1}{V^{\frac{1-r}{2}}}\left( {{0}} \right) + \theta {\lambda _2}}}{{\theta {\lambda _2}}}} \right\}.
\end{aligned}
\end{equation}

For $t \ge {T}$, the tracking error can arrive at
\begin{equation}
\begin{aligned}
\left| {{z_1}} \right| &\le \left| {{v_1}} \right| + \left| {{\xi _1}}\right| \\
&\le \min \left\{ {2\sqrt {\frac{{2{\lambda _3}}}{{\left( {1 - \theta } \right){\lambda _1}}}} ,2\sqrt {2{{\left( {\frac{{{\lambda _3}}}{{\left( {1 - \theta } \right){\lambda _2}}}} \right)}^{\frac{2}{{1 + r}}}}} } \right\}.
\end{aligned}
\end{equation}

The proof is completed.
\end{proof}

\begin{remark}
The positive fractional power of less than $1$ in the virtual control law may lead to singularity when calculating the time derivative of the virtual control signal, which is eliminated in (43) by leveraging our newly proposed inequality. 
\end{remark}

\begin{remark}
It is known from (55) that the convergent time of the attitude tracking error partially depends on $t_1$ and $t_2$, which are determined by the convergent rates of the disturbance observer and command filter. Therefore, based on the proof of \emph{Theorem 1} and the analysis of \cite{lemma3}, the integration of ASDO and FFTCF into the proposed control strategy enables the closed-loop system to achieve finite-time convergence with faster response. 
\end{remark}
\begin{figure}[htb]
	\centering
    \subfigure[The disturbance $d_1$ and its estimation]{
	\includegraphics[width=0.4\textwidth]{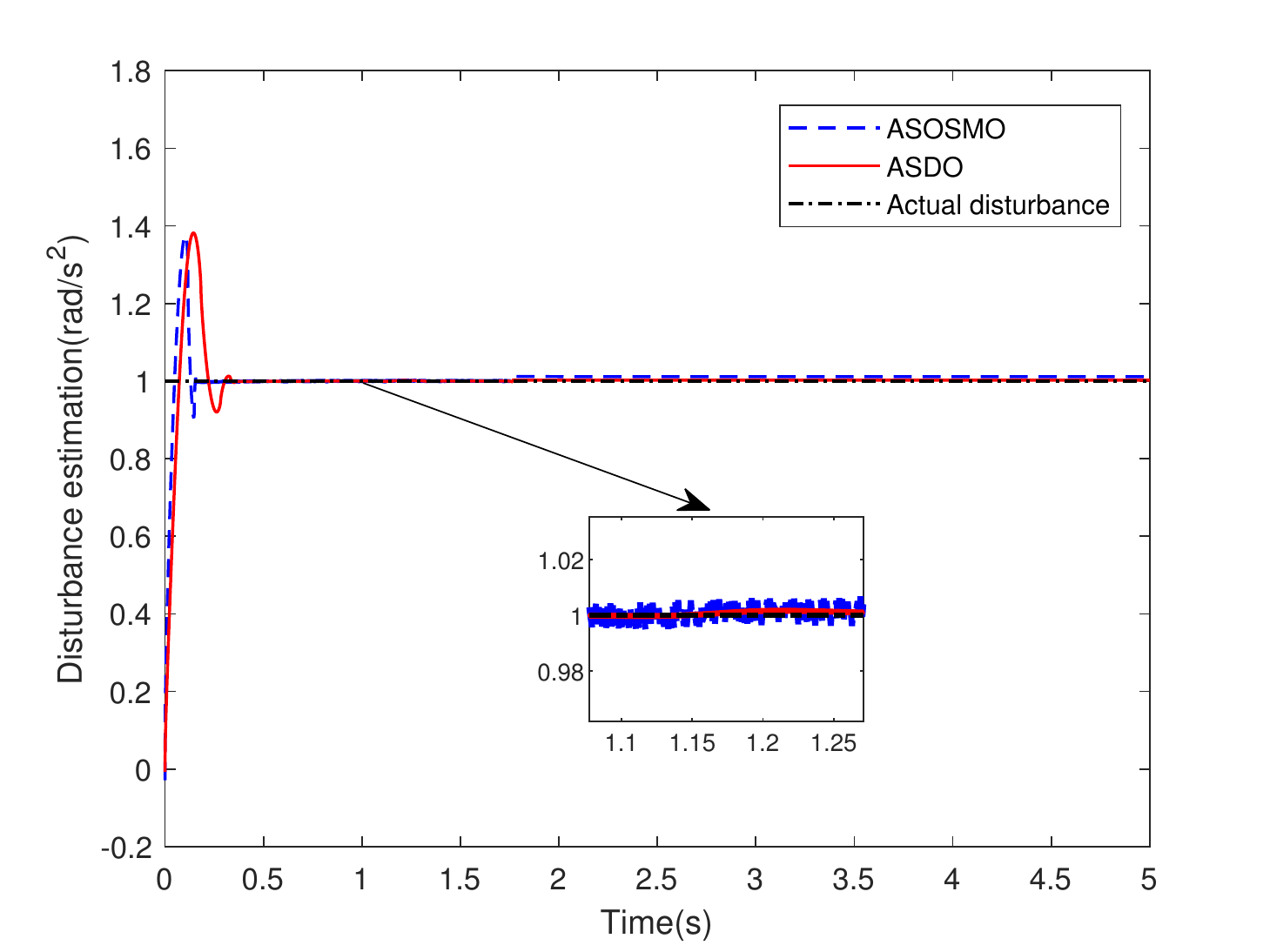}}
	\subfigure[Observer error]{
	\includegraphics[width=0.4\textwidth]{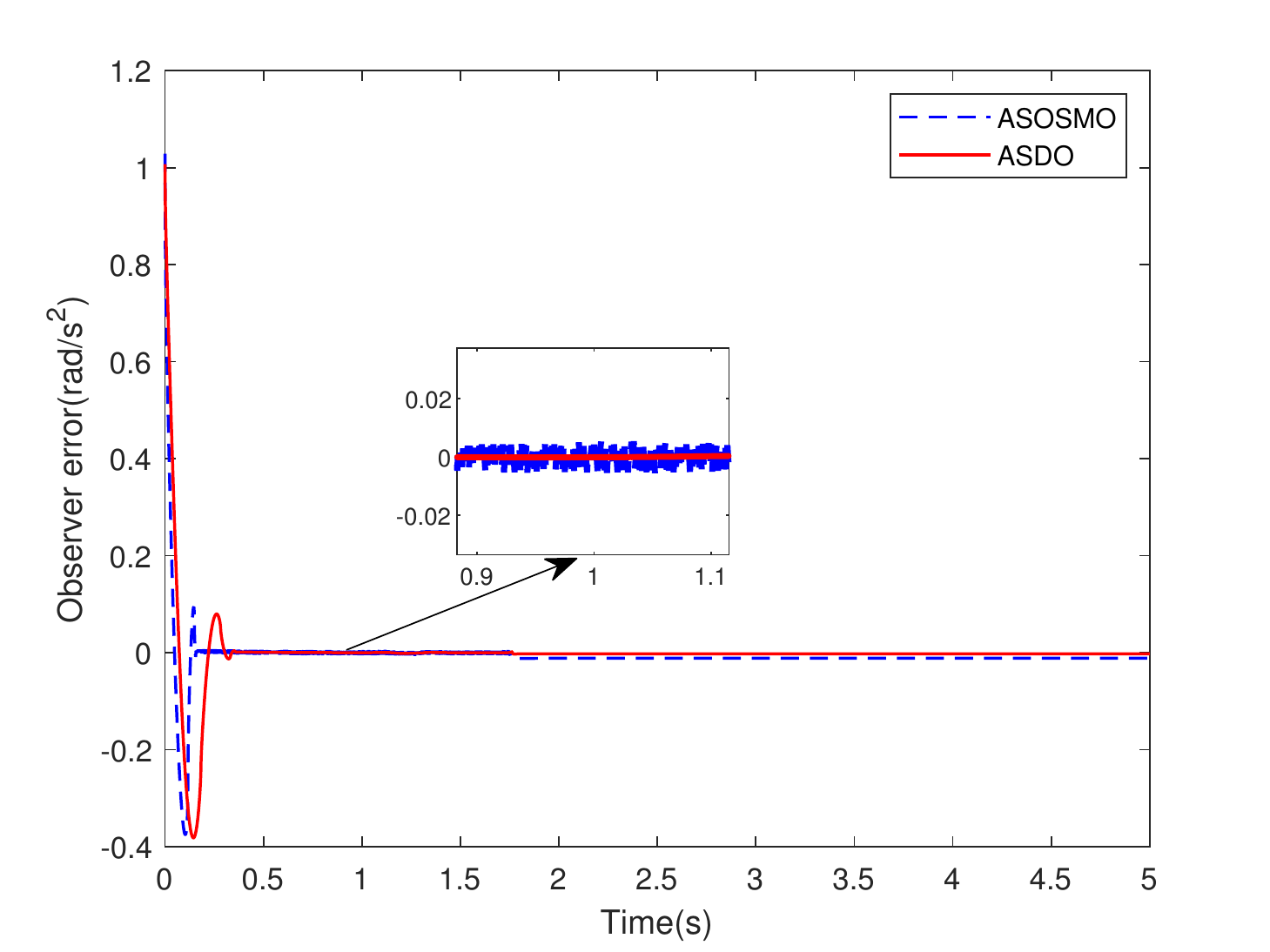}}
	\caption{Results of constant disturbance estimation (elevation channel)}
\end{figure}
\begin{figure}[htb]
	\centering
    \subfigure[The disturbance $d_1$ and its estimation]{
	\includegraphics[width=0.4\textwidth]{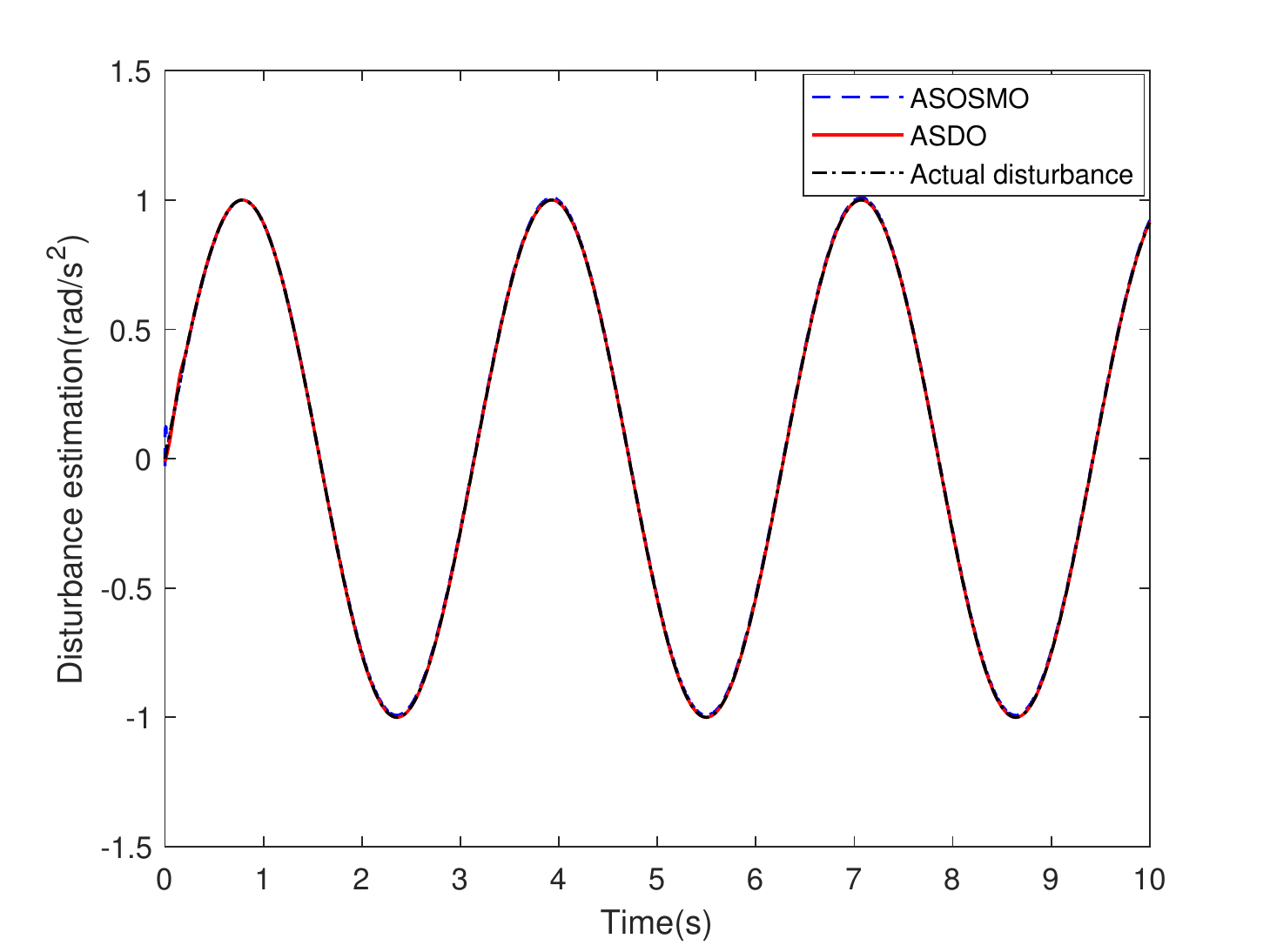}}
	\subfigure[Observer error]{
	\includegraphics[width=0.4\textwidth]{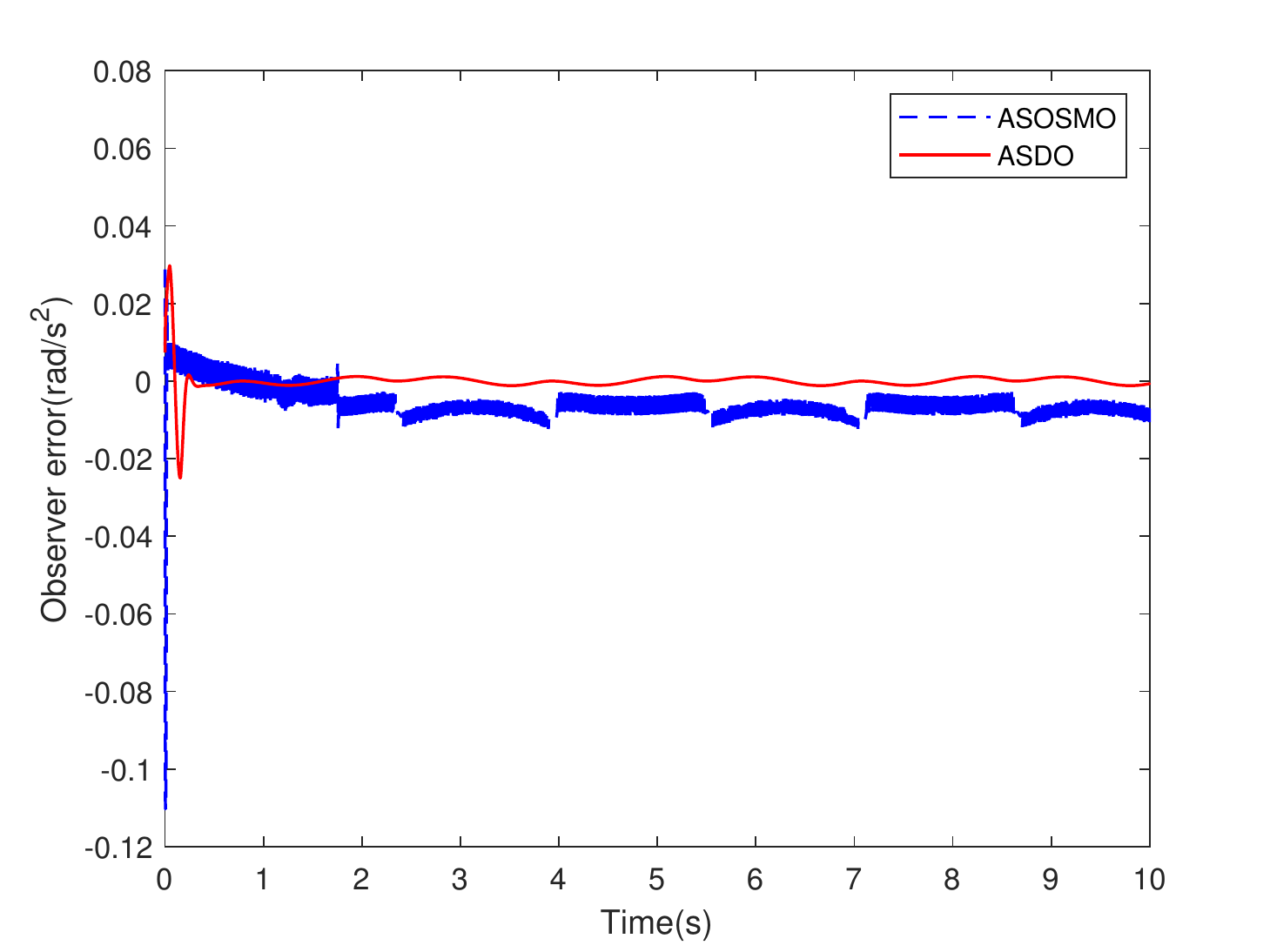}}
	\caption{Results of time-varying disturbance estimation (elevation channel)}
\end{figure}
\begin{figure}[htb]
	\centering
    \subfigure[The disturbance $d_2$ and its estimation]{
	\includegraphics[width=0.4\textwidth]{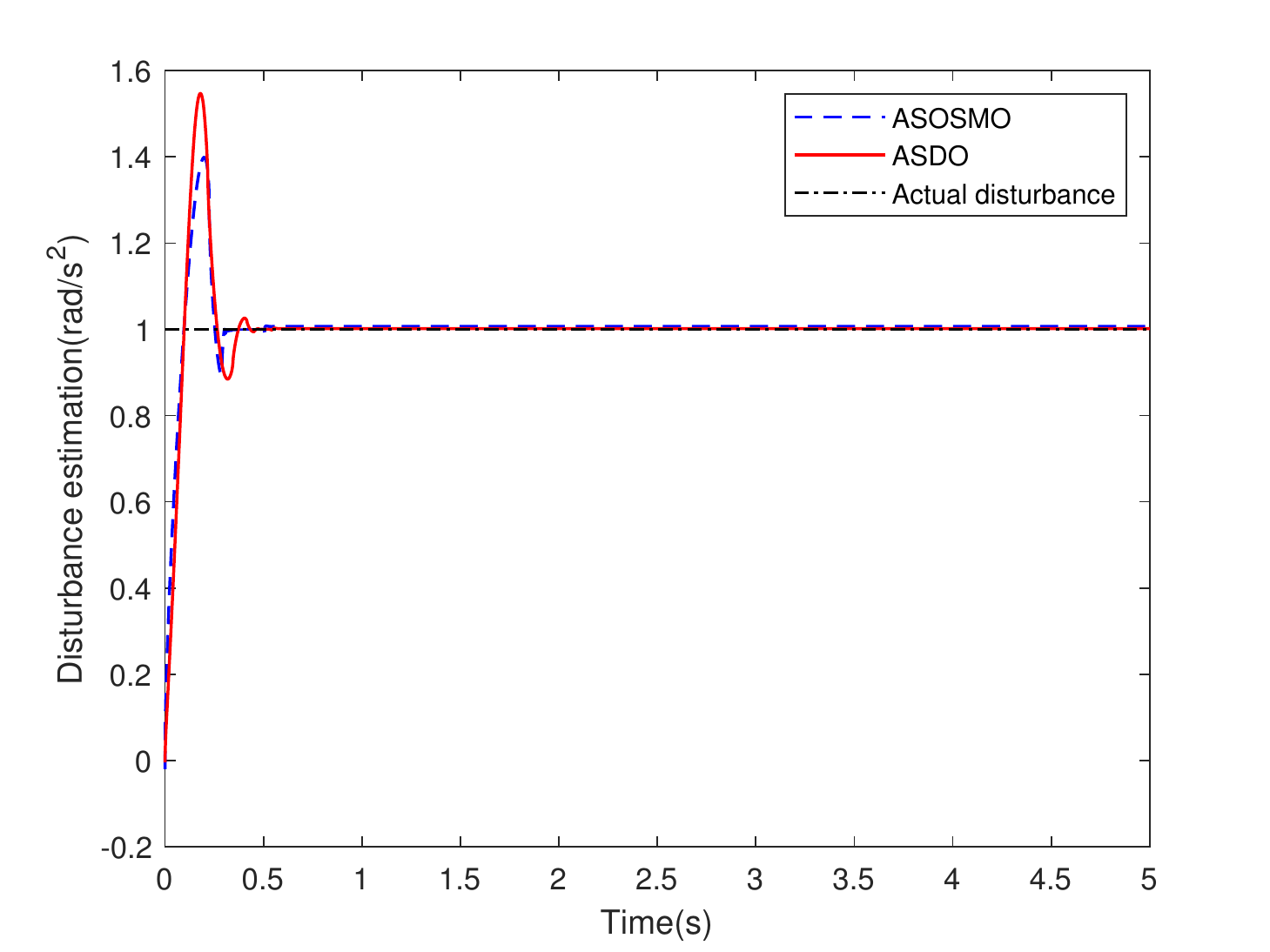}}
	\subfigure[Observer error]{
	\includegraphics[width=0.4\textwidth]{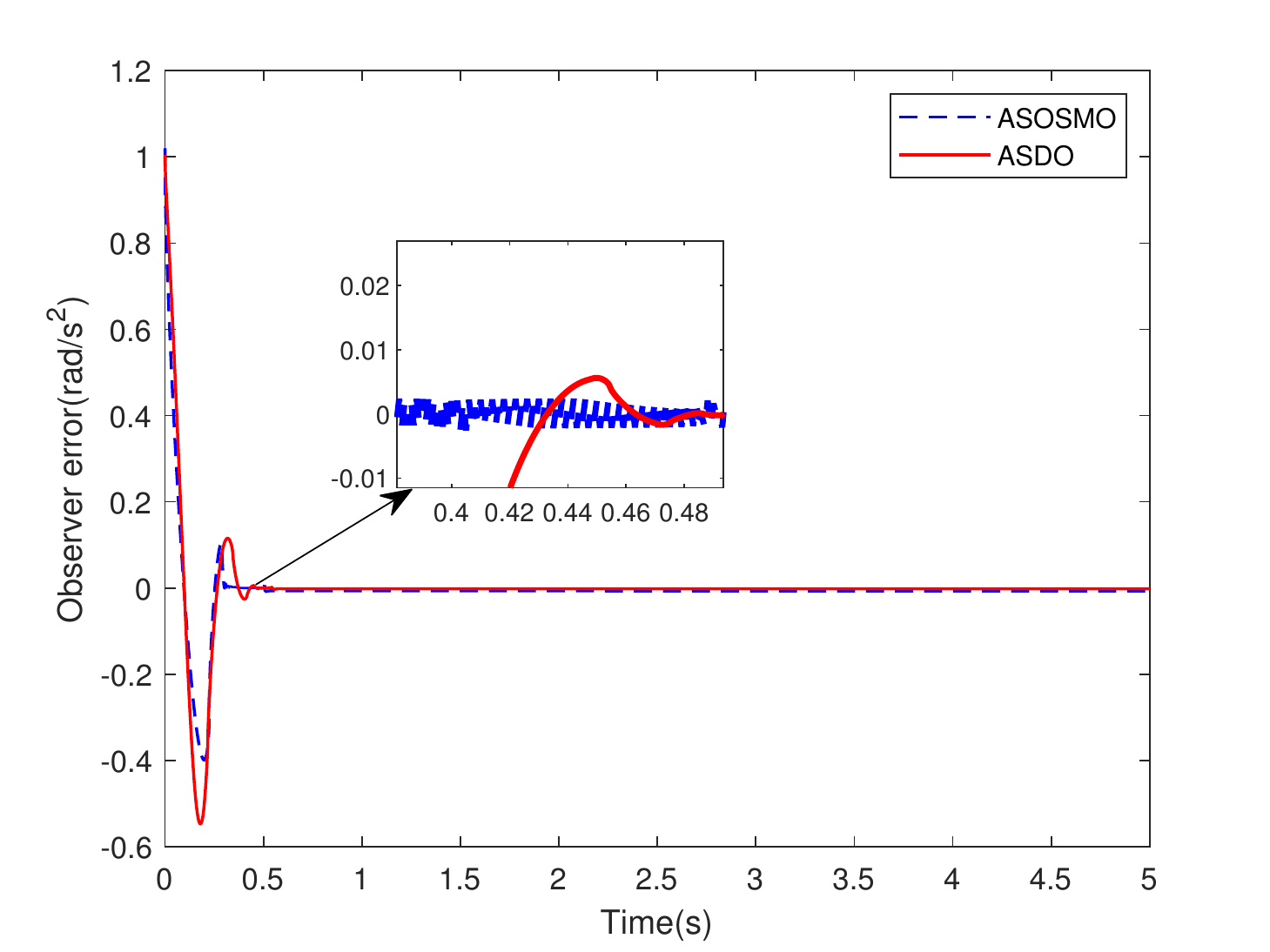}}
	\caption{Results of constant disturbance estimation (pitch channel)}
\end{figure}
\begin{figure}[htb]
	\centering
    \subfigure[The disturbance $d_2$ and its estimation]{
	\includegraphics[width=0.4\textwidth]{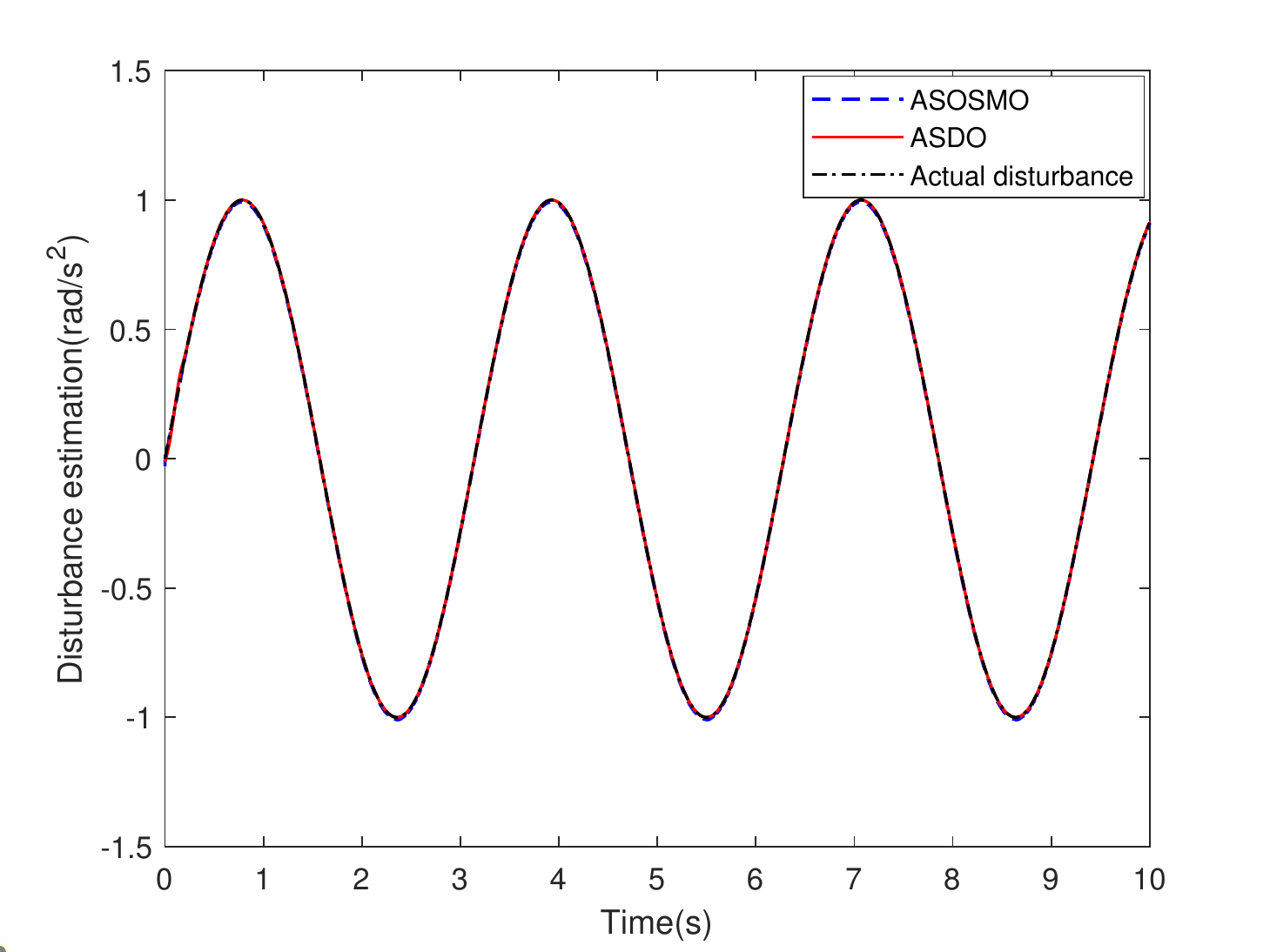}}
	\subfigure[Observer error]{
	\includegraphics[width=0.4\textwidth]{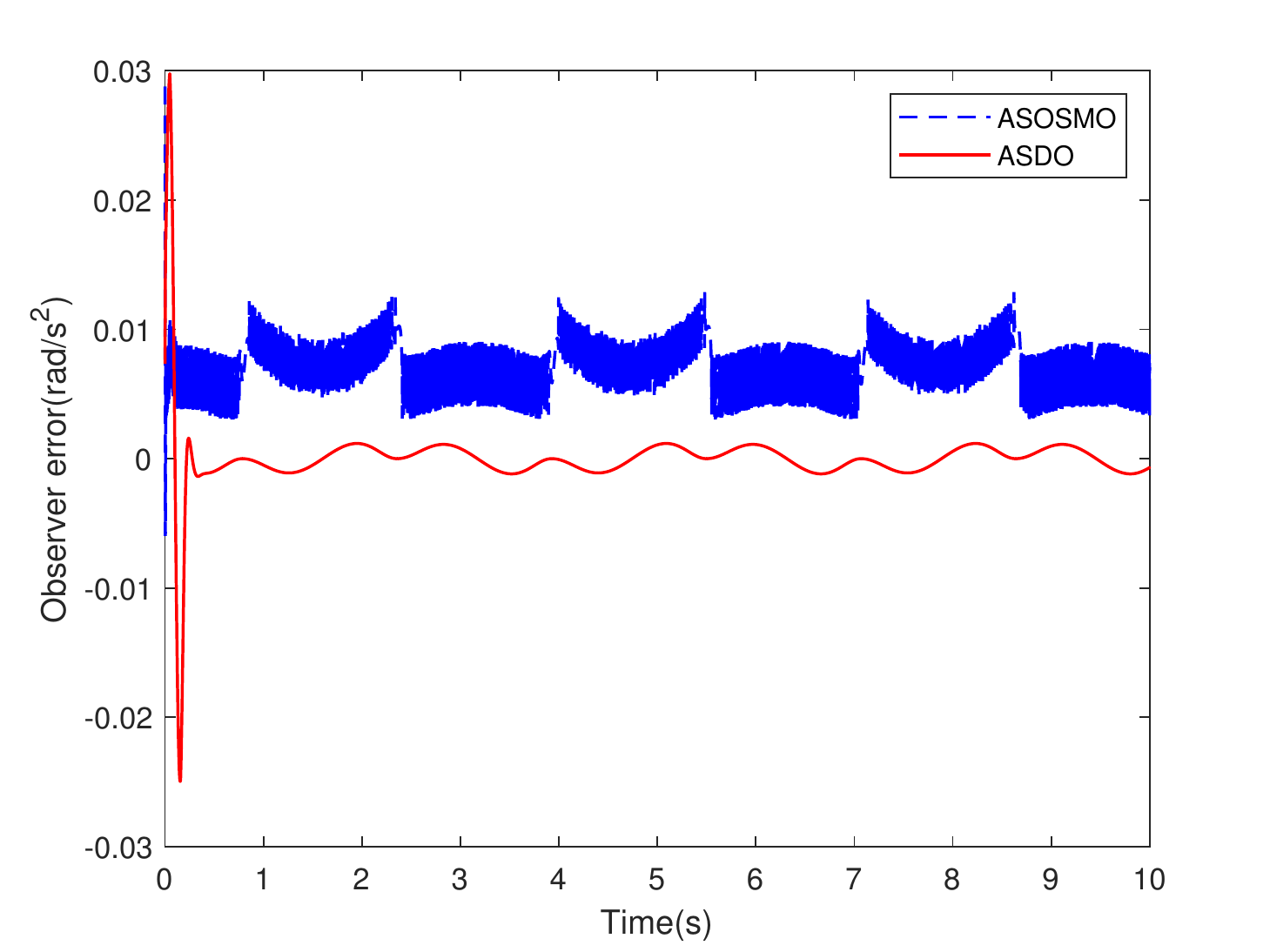}}
	\caption{Results of time-varying disturbance estimation (pitch channel)}
\end{figure}
\begin{figure}[htb]
	\centering
    \subfigure[The control input $V_f$]{
	\includegraphics[width=0.4\textwidth]{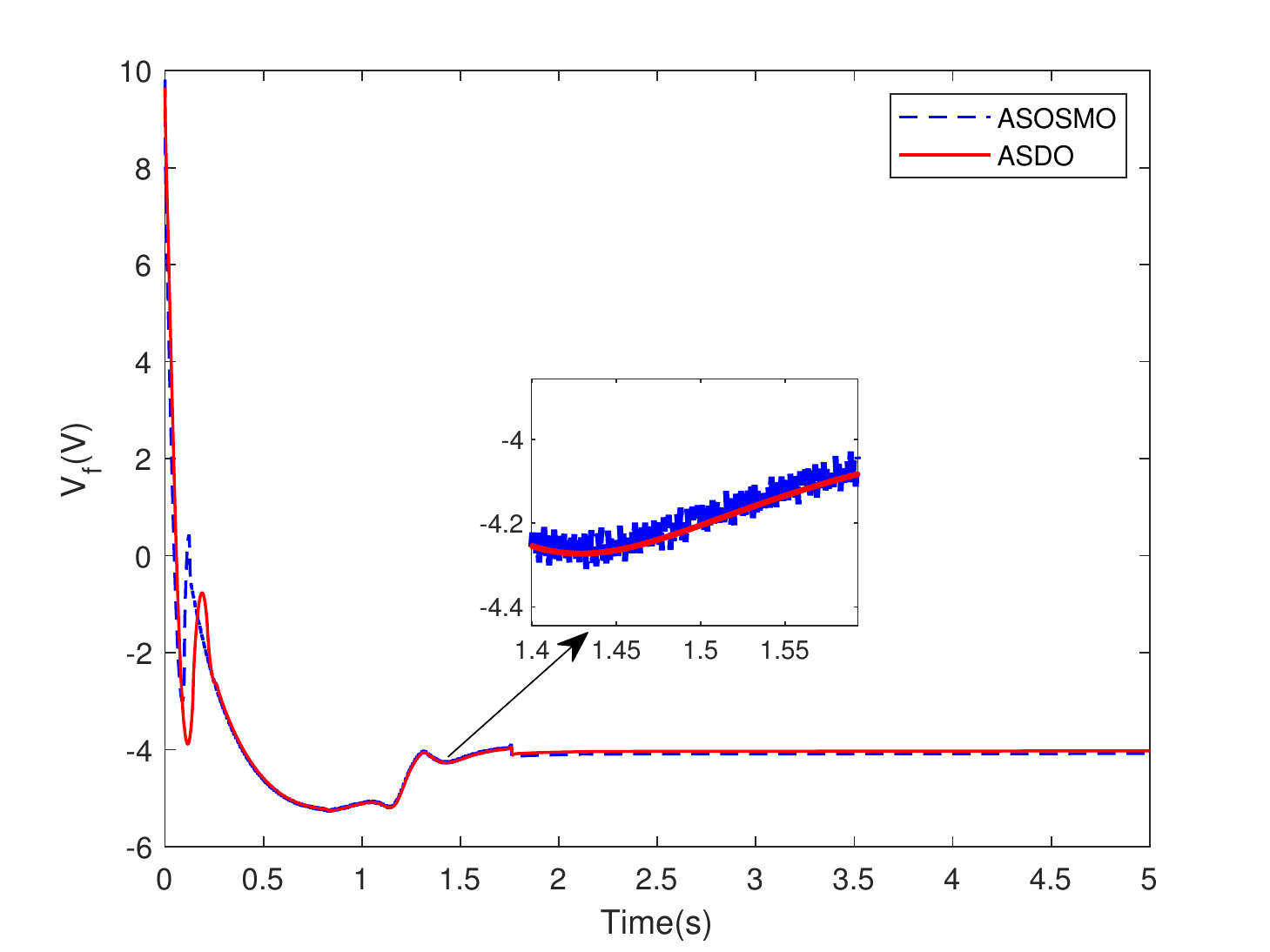}}
	\subfigure[The control input $V_b$]{
	\includegraphics[width=0.4\textwidth]{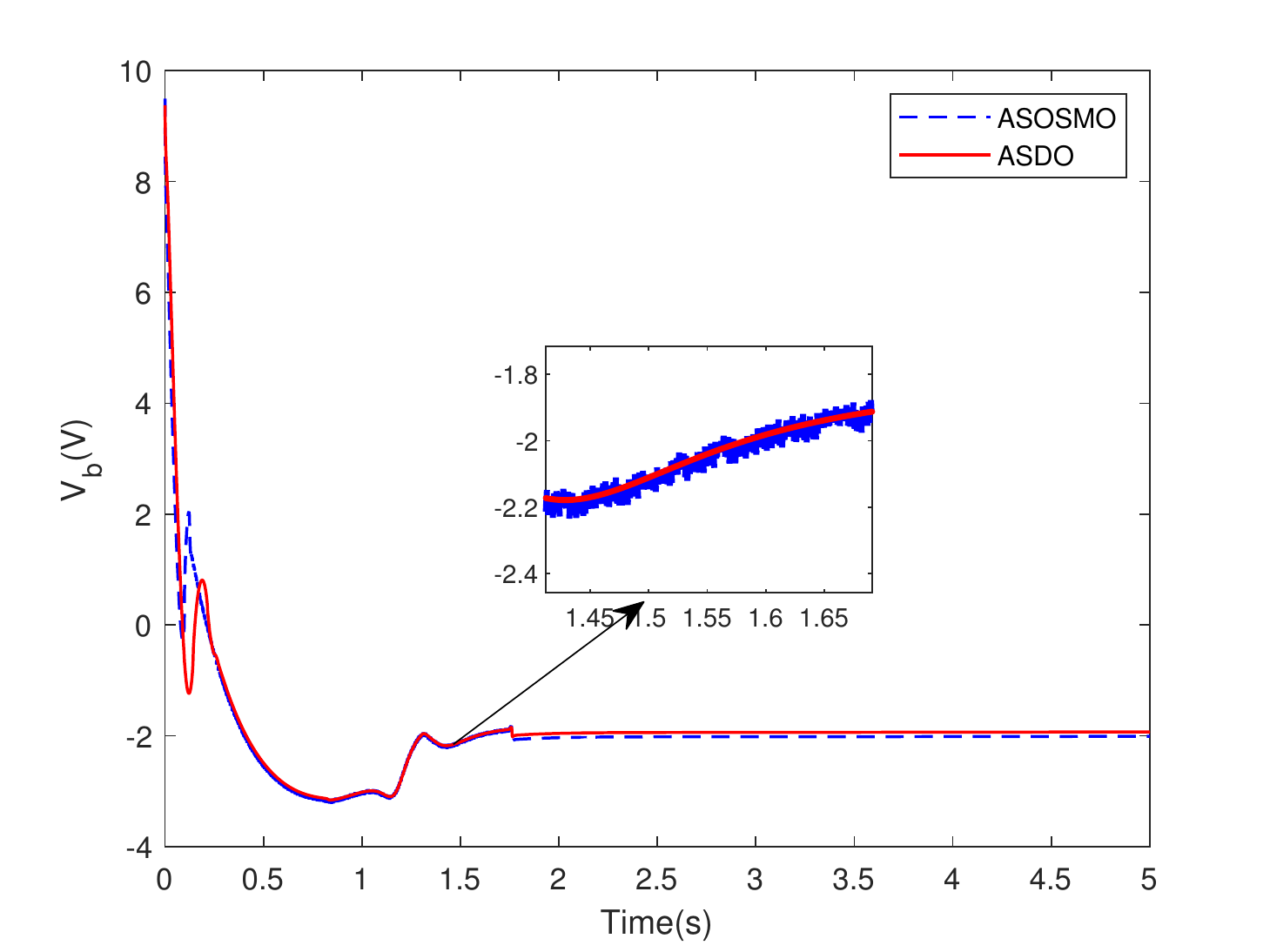}}
	\caption{Results of control inputs (constant disturbance)}
\end{figure}
\section{Simulation Results}
This section presents two sets of contrastive numerical simulations to expound the effectiveness and superiority of our proposed control scheme. Table \ref{table:1} shows the parameter values of the helicopter system, and the fixed sampling time for all the simulation experiments is set to 0.001 second. In these experiments, the initial angle of the elevation channel is ${x_1}\left( 0 \right) =  - \frac{{2\pi }}{{15}}(rad)$, and the pitch channel is ${x_3}\left( 0 \right) =  0 (rad)$. The desired trajectories are given as
\begin{equation}
{x_{1d}}(t){\rm{ = }}-0.2\cos (0.08t) - 0.1,{x_{3d}}(t) = 0.1\sin (0.06t).
\end{equation}

The purposes of two group experiments are outlined as follows

(\romannumeral1) The first experiment (Case I) illustrates the superior performance of our designed ASDO in estimating constant and time-varying disturbances by comparing with the ASOSMO proposed in \cite{ASTC2015}.

(\romannumeral2) The second experiment (Case II) demonstrates the merits of our proposed control scheme in the attitude tracking of the helicopter system under lumped disturbances via contrastive numerical simulations. To better display simulation comparison results, a time-varying disturbance is adopted.
\subsection{Case \uppercase\expandafter{\romannumeral1}: Performance analysis of ASDO}
In Case \uppercase\expandafter{\romannumeral1}, the parameters of ASDO are given as ${k_1}=2,{k_2}=2.5,{k_3}=4,{k_4}=30$, and the parameters setting of elevation channel controller are ${\bar k_1} = 1,{\bar k_2} = 2,r = 0.6,{l_1} = {l_2} =1,{s_1} = {s_2} = 0.5,\varepsilon_c  = 0.01,{a_0} = 5,{a_1} = 0.5,{b_0} = 2,{b_1} = 0.5,{\gamma_3} = {\gamma_4} = 0.5,\sigma_p  = 0.1,q = 30,\eta  = 1,\varepsilon_p  = 0.1$. The parameters of pitch channel controller are given as ${\bar k_1} = 3,{\bar k_2} = 5,{l_1} = {l_2} =2,{s_1} = {s_2} = 2$, while the other parameters of the pitch channel controller are the same as those of the elevation channel controller. The only difference in the comparison simulation is that the ASOSMO in \cite{ASTC2015} is adopted to replace the ASDO. The constant and time-varying disturbances are assigned as ${d_1} = d_2 = 1({{rad} \mathord{\left/
 {\vphantom {{rad} {{s^2}}}} \right.
 \kern-\nulldelimiterspace} {{s^2}}}), {d_1} = d_2 = \sin \left( {2t} \right)({{rad} \mathord{\left/
 {\vphantom {{rad} {{s^2}}}} \right.
 \kern-\nulldelimiterspace} {{s^2}}})$, respectively.

\begin{figure}[htb]
	\centering
    \subfigure[The control input $V_f$]{
	\includegraphics[width=0.4\textwidth]{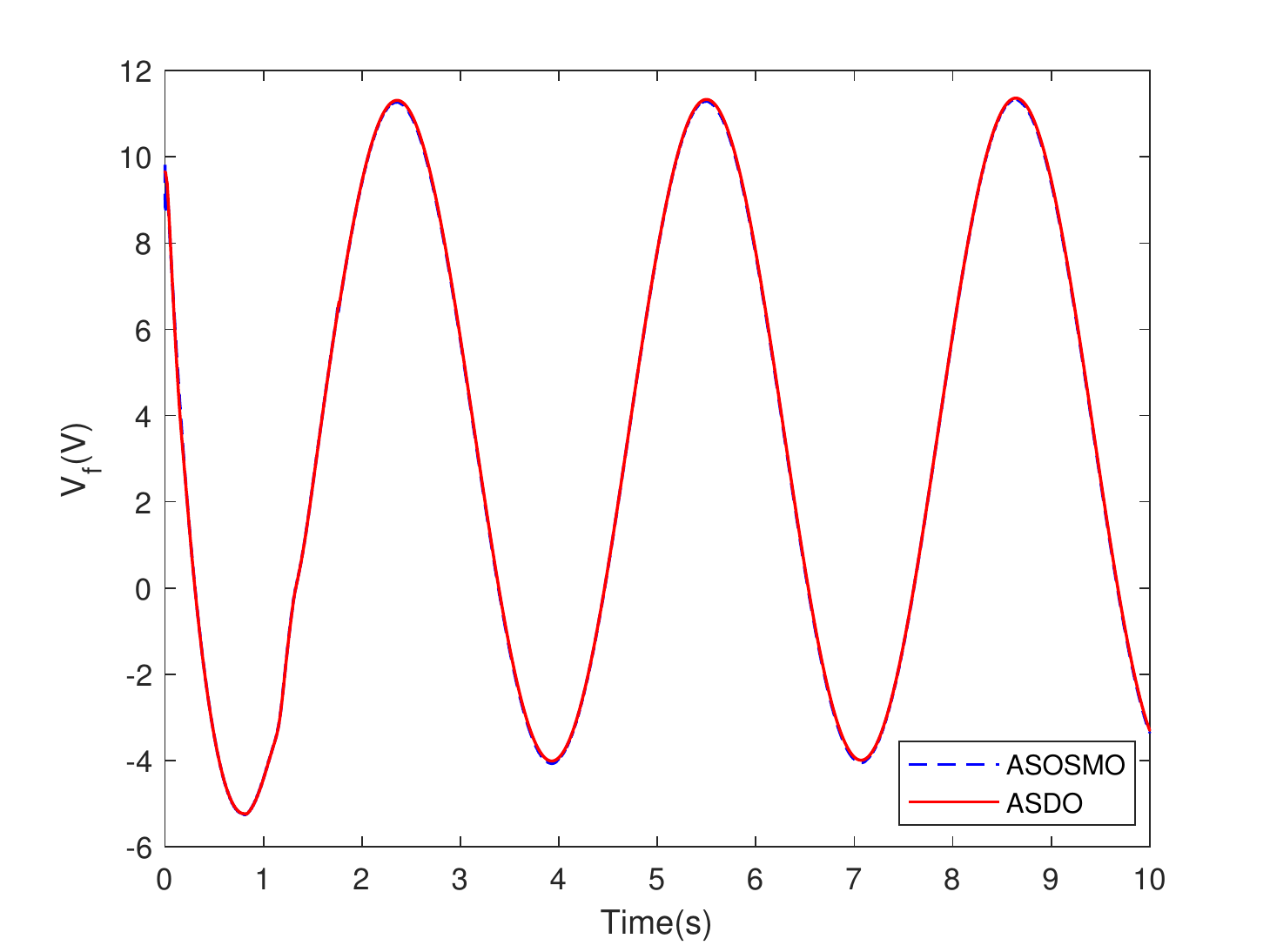}}
	\subfigure[Partial enlarged graph of $V_f$]{
	\includegraphics[width=0.4\textwidth]{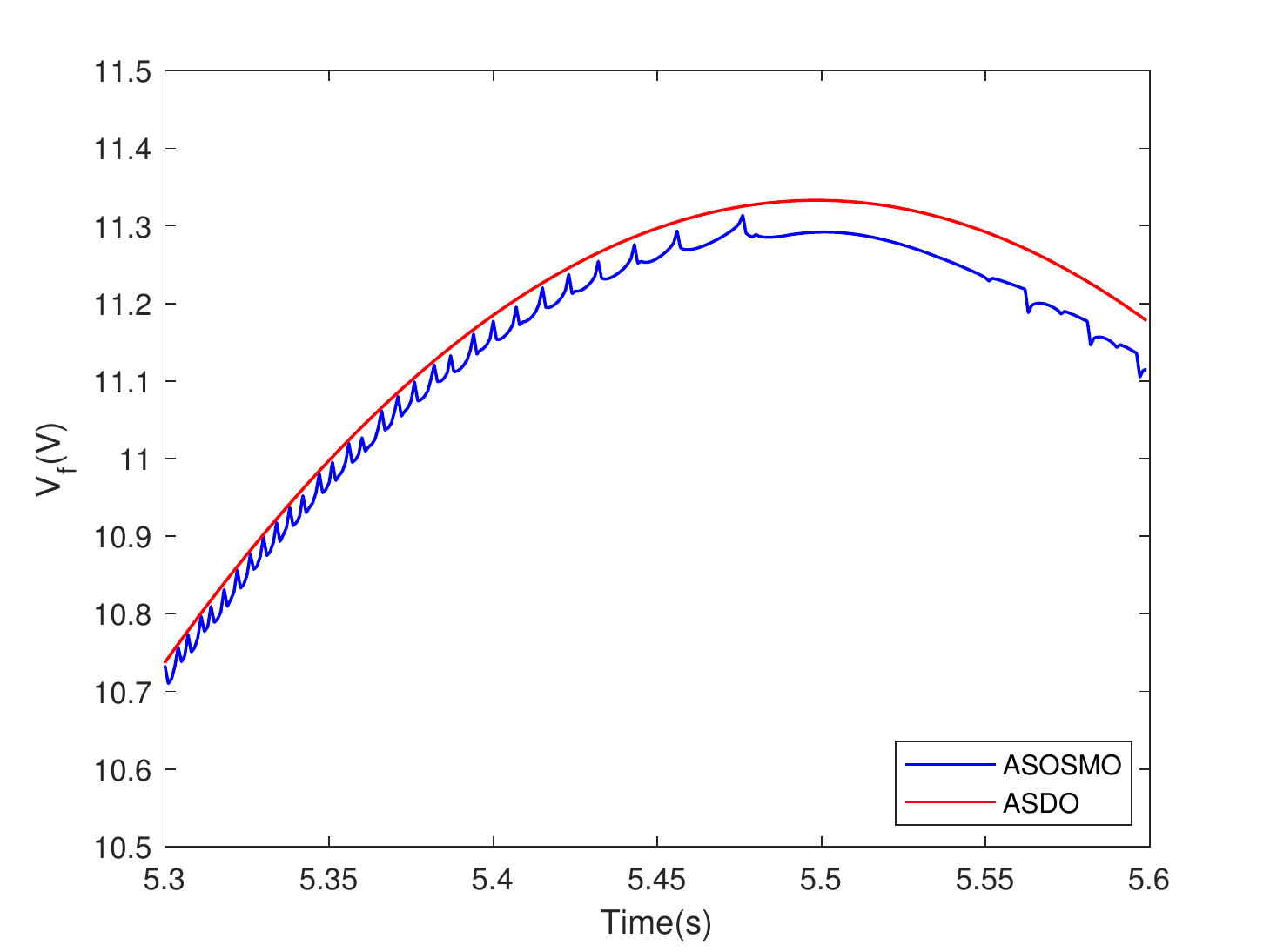}}
    \subfigure[The control input $V_b$]{
	\includegraphics[width=0.4\textwidth]{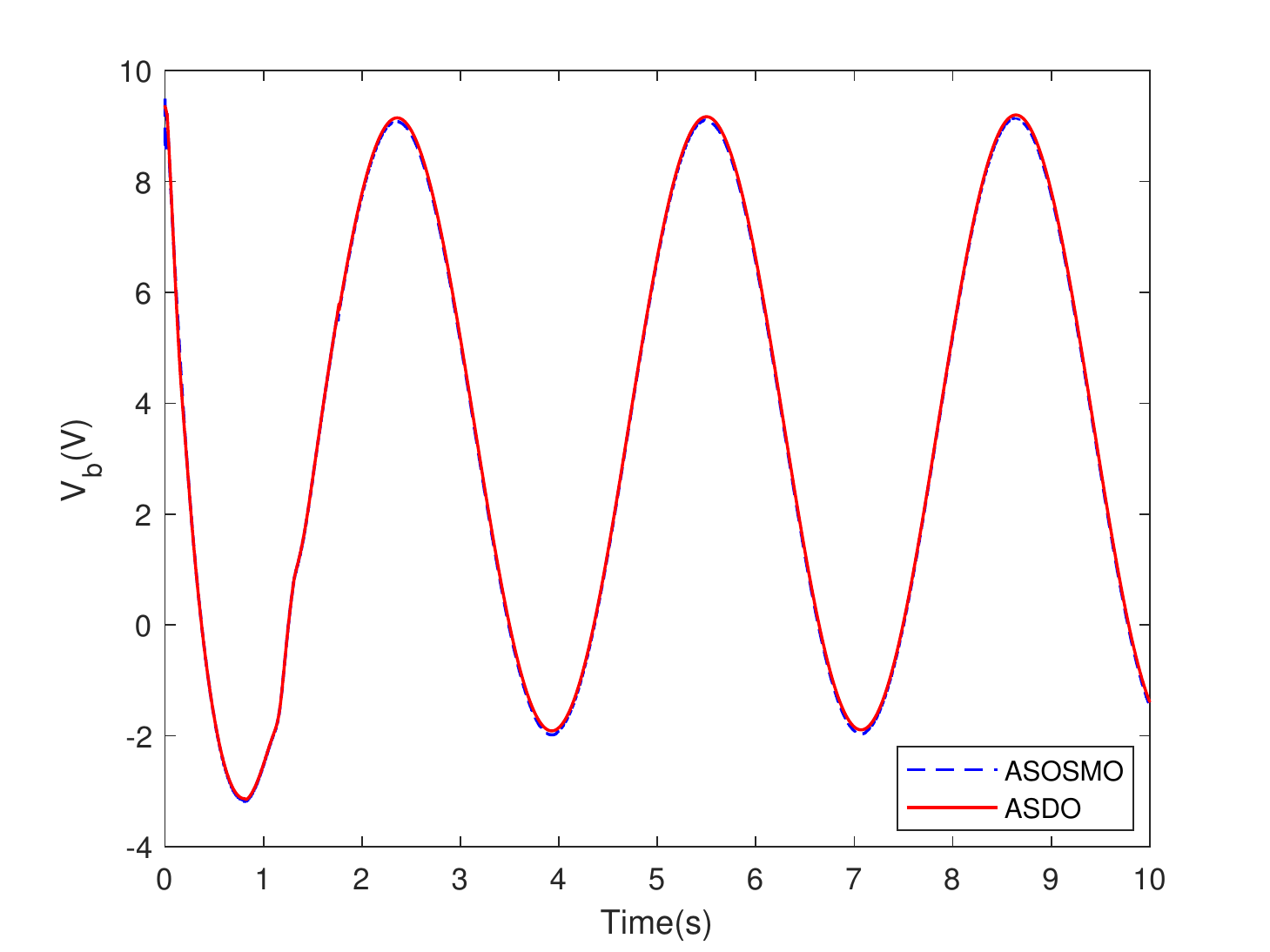}}
	\subfigure[Partial enlarged graph of $V_b$]{
	\includegraphics[width=0.4\textwidth]{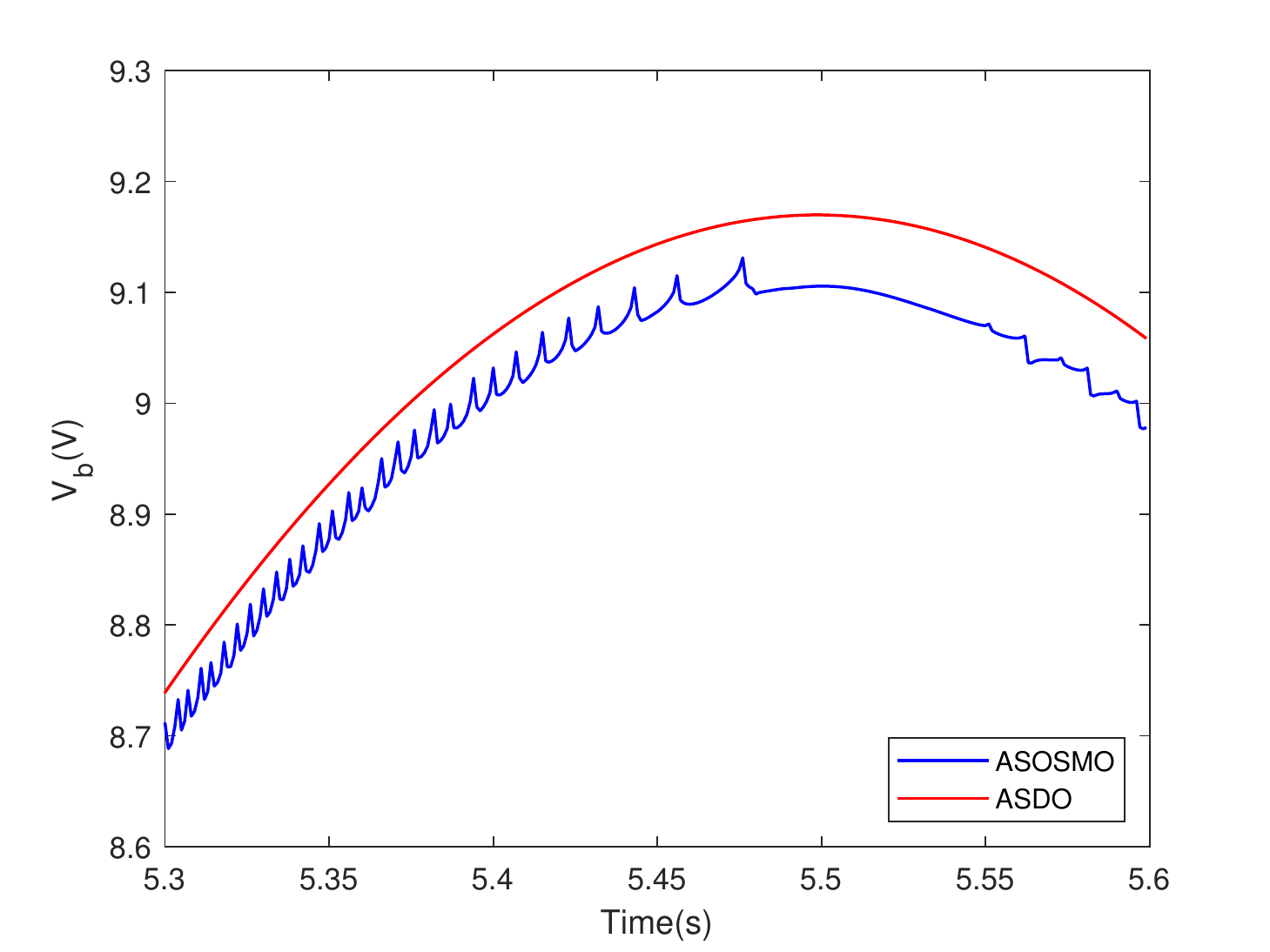}}
	\caption{Results of control inputs (time-varying disturbance)}
\end{figure}
\begin{figure}[htb]
	\centering
    \subfigure[Tracking response of elevation angle]{
	\includegraphics[width=0.3\textwidth]{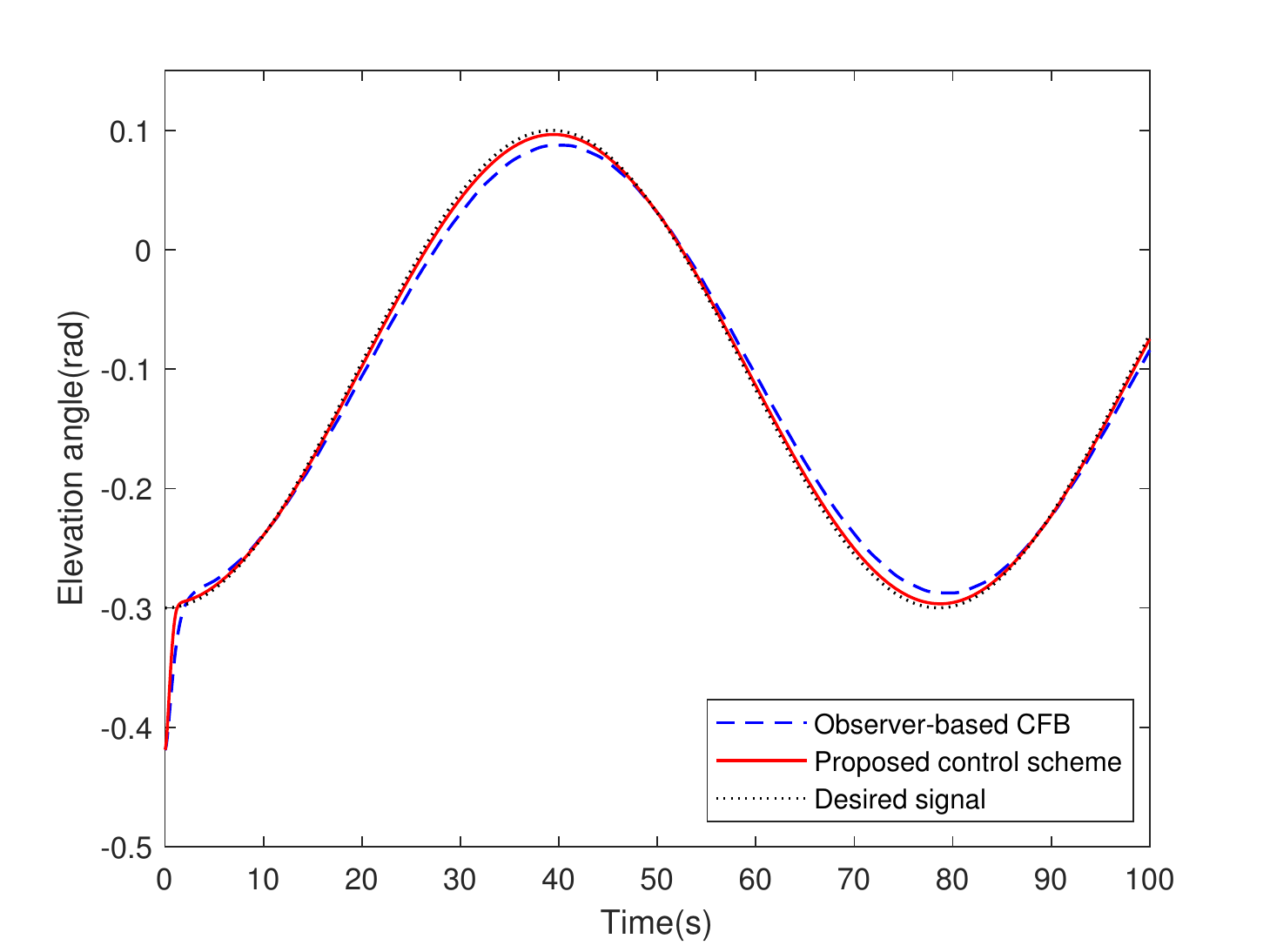}}
	\subfigure[Tracking error of elevation angle]{
	\includegraphics[width=0.3\textwidth]{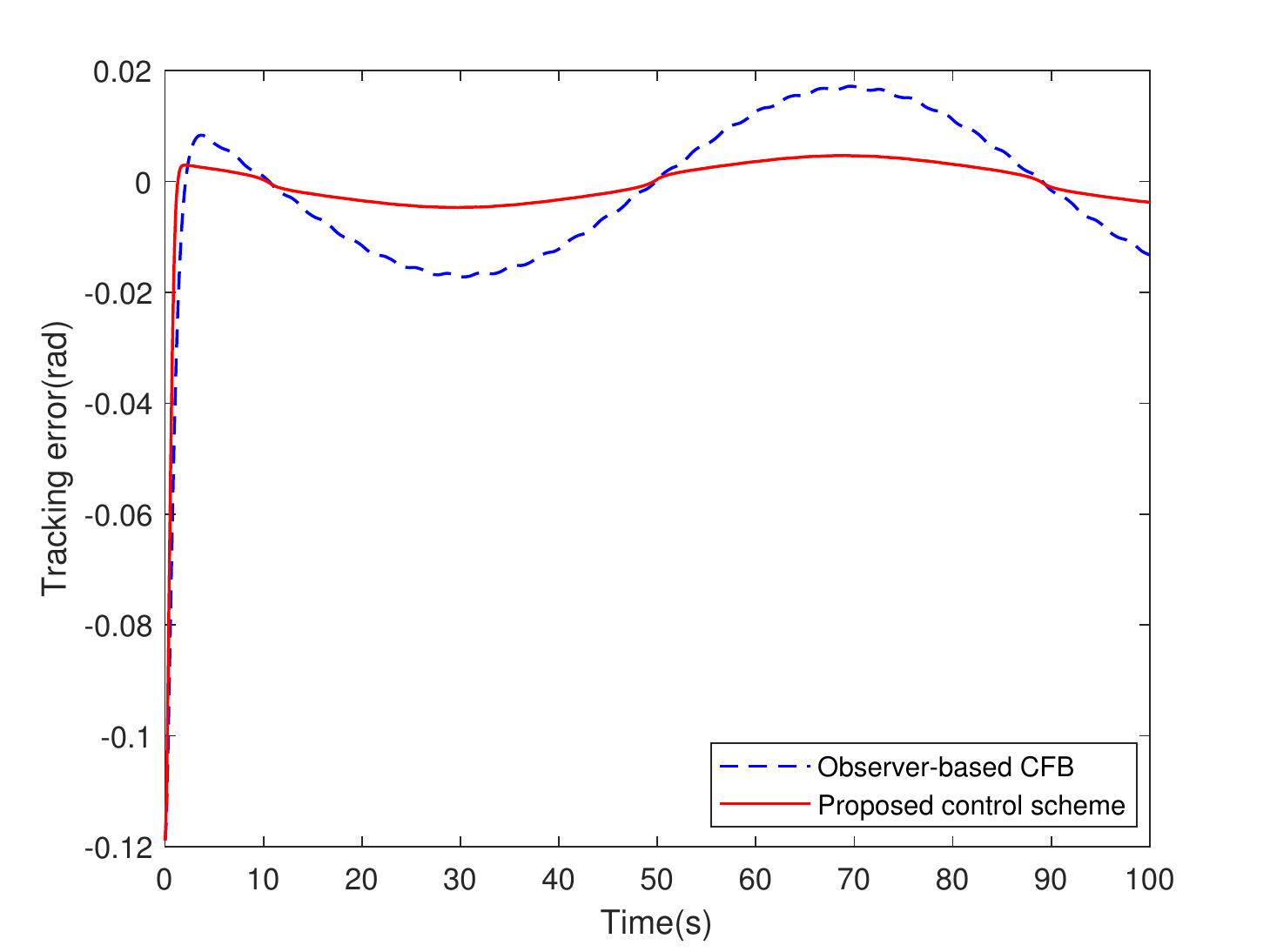}}
	\subfigure[Partial enlarged graph of tracking error]{
	\includegraphics[width=0.3\textwidth]{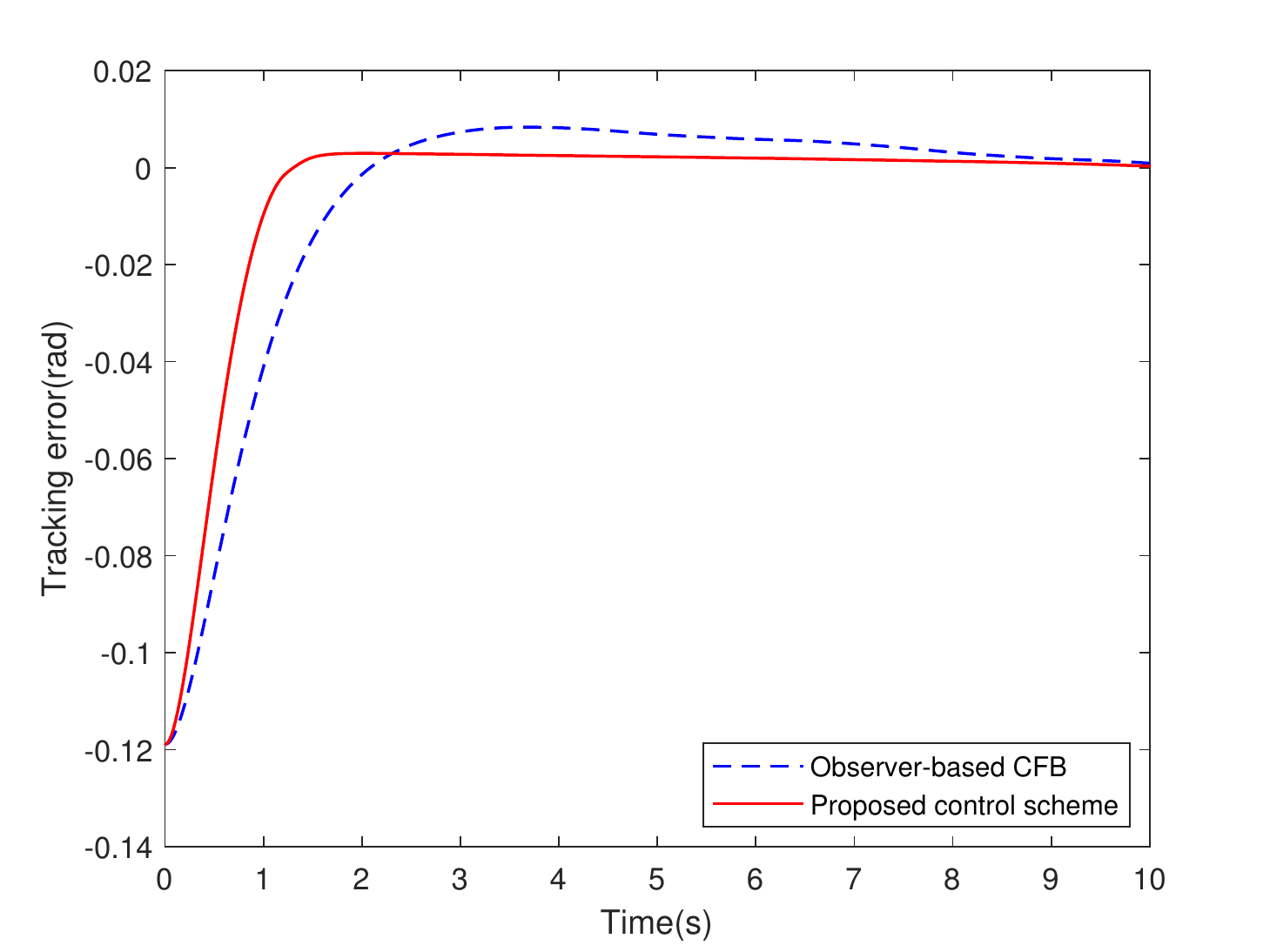}}
	\caption{Results of elevation angle tracking}
\end{figure}
\begin{figure}[htb]
	\centering
    \subfigure[Tracking response of pitch angle]{
	\includegraphics[width=0.3\textwidth]{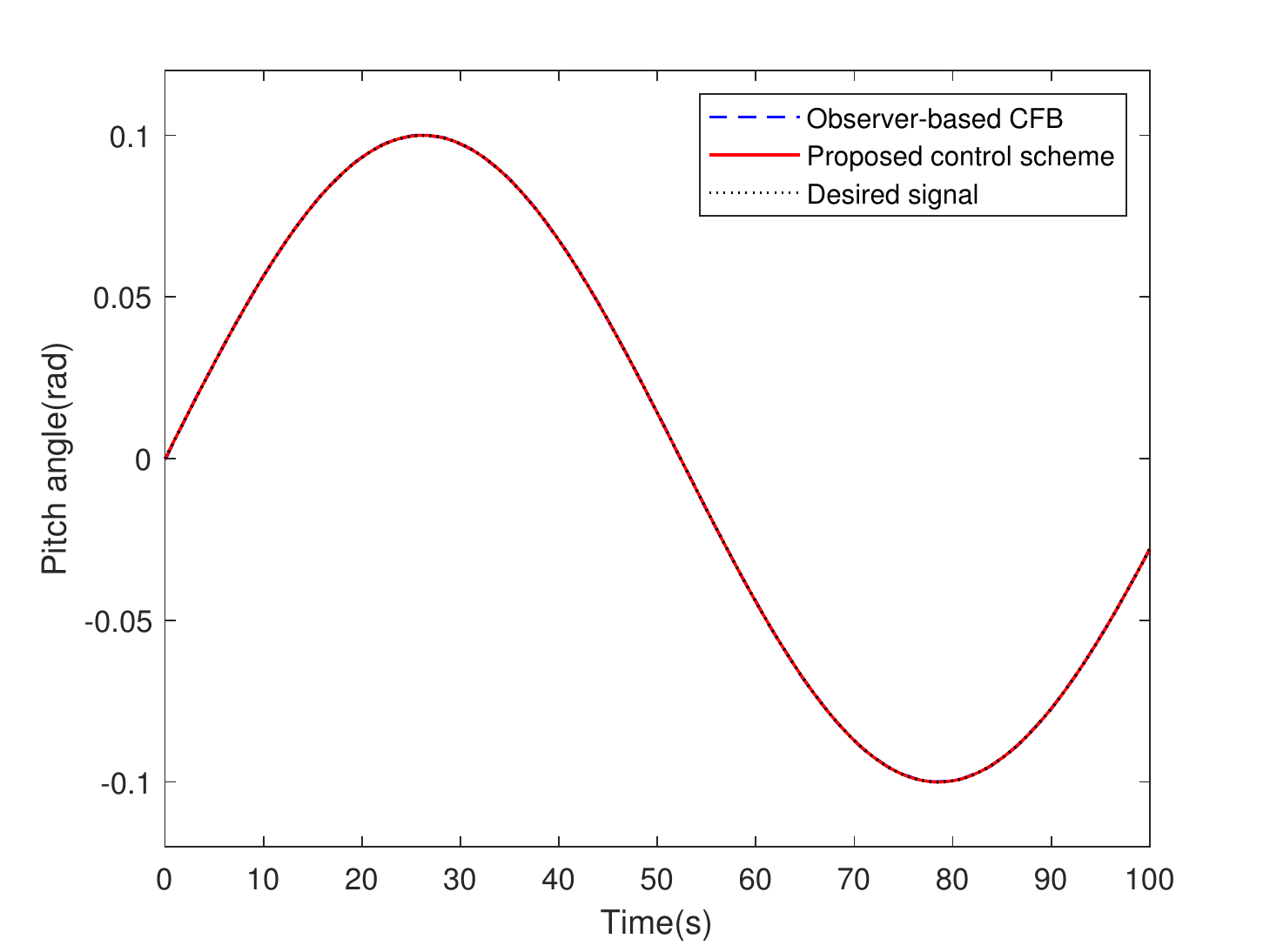}}
	\subfigure[Tracking error of pitch angle]{
	\includegraphics[width=0.3\textwidth]{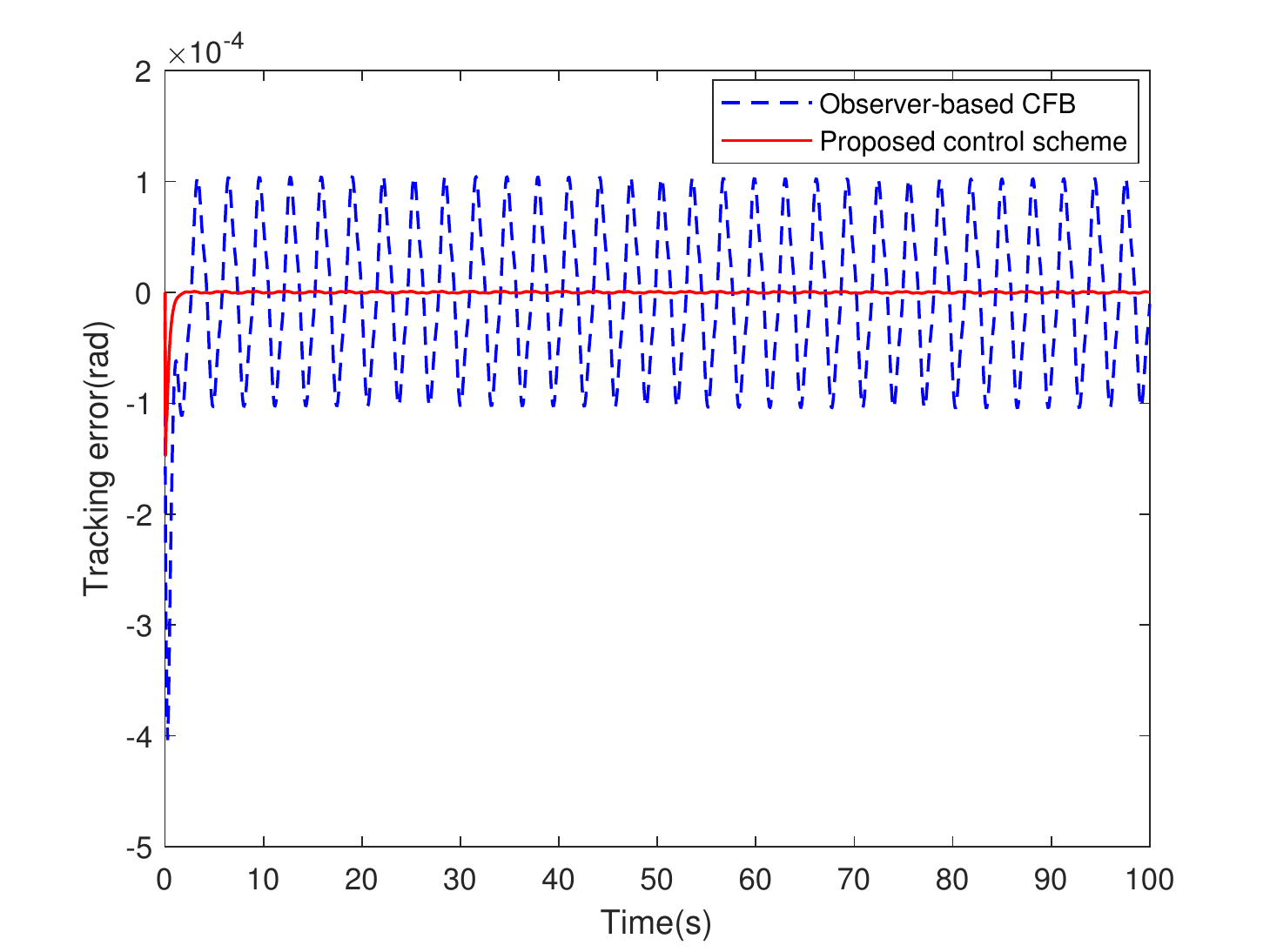}}
	\subfigure[Partial enlarged graph of tracking error]{
	\includegraphics[width=0.3\textwidth]{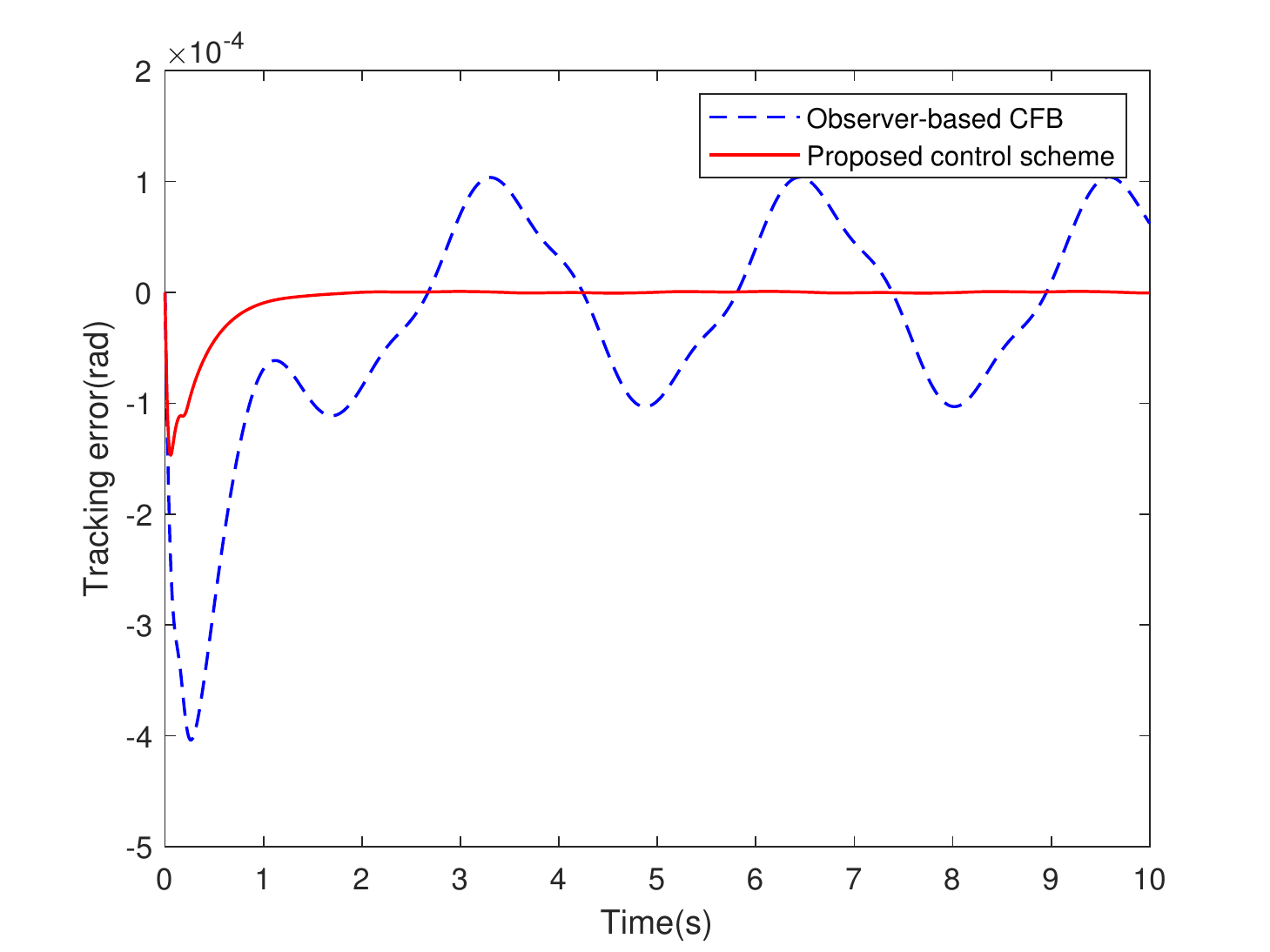}}
	\caption{Results of pitch angle tracking}
\end{figure}

Simulation results are given in Figures 2 to 7. Figures 2 to 5: (a) show the constant and time-varying disturbances of corresponding control channels and their estimations via ASDO and ASOSMO, respectively. Figures 2 to 5: (b) display the estimation errors of corresponding disturbance observers. Figure 6 and Figure 7 show the control inputs.

From Figures 2 to 5: (b), one can obtain that the estimation error of constant disturbance fast converges to the origin in finite time, and the approximation error of time-varying disturbance fast converges to a neighborhood of the origin in finite time by utilizing our designed ASDO. In addition, the ASDO maintains the fast finite-time convergence and adaptability to disturbance of ASOSMO, while having a smoother output than ASOSMO. Furthermore, the contrastive simulation results in Figure 6 and Figure 7 illustrate that the smooth outputs of ASDO lead to smooth control inputs of the helicopter system.

\subsection{Case \uppercase\expandafter{\romannumeral2}: Attitude tracking control with time-varying lumped disturbances}
In Case \uppercase\expandafter{\romannumeral2}, the parameters settings of elevation and pitch channel controllers are the same as Case \uppercase\expandafter{\romannumeral1}.  The CFB approach in \cite{11.li2021} combined with the ASDO is employed as the comparison. The time-varying disturbances are set as ${d_1} = d_2 = \sin \left( {2t} \right)({{rad} \mathord{\left/{\vphantom {{rad} {{s^2}}}} \right.\kern-\nulldelimiterspace} {{s^2}}})$.

Simulation results are given in Figure 8 and Figure 9. Figure 8: (a) and Figure 9: (a) display the tracking responses of elevation and pitch angles via the designed control scheme and observer-based CFB approach, respectively. Figure 8: (b) and Figure 9: (b) show the tracking errors of corresponding control channels. Figure 8: (c) and Figure 9: (c) present the partial enlargement of corresponding tracking errors to show the tracking error responses more clearly.

From Figure 8: (b) and Figure 9: (b), it is observed that the attitude tracking errors fast converge to a small neighborhood of the origin in finite time by utilizing the proposed control scheme. Moreover, the comparative simulation results in Figure 8: (c) and Figure 9: (c) demonstrate that the designed controller not only has a faster convergent rate, but also achieves better tracking performance than the observer-based CFB approach.

\section{Conclusion}
In this paper, a novel ASDO-based fast finite-time adaptive backstepping control scheme has been presented to address the attitude tracking control problem of a small unmanned helicopter system in the presence of lumped disturbances. The ASDO was designed to estimate lumped disturbances. Based on our newly presented inequality, a novel virtual control law was developed to suppress the potential singularity caused by the time derivative of the virtual control law. By introducing a FFTCF with auxiliary dynamic system into finite time backstepping control protocol, the ``explosion of complexity" was handled and the impact of filter error was diminished. To further enhance the tracking performance, an adaptive law with $\sigma $-modification term was designed to attenuate the ASDO approximation error. It is proved that the attitude tracking errors fast converge to a sufficiently small region of the origin in finite time. The contrastive simulation study was executed to illustrate the effectiveness and superiority of the proposed control scheme.

\section*{Acknowledgment}
This work was supported partially by the Major Scientific and Technological Special Project of Heilongjiang Province under Grant 2021ZX05A01, the Heilongjiang Natural Science Foundation under Grant LH2019F020, and the Major Scientific and Technological Research Project of Ningbo under Grant 2021Z040.
\bibliographystyle{elsarticle-num}
\bibliography{myRef}

\end{document}